\documentclass[a4paper,UKenglish,cleveref, autoref, thm-restate]{lipics-v2021}
\hideLIPIcs
\nolinenumbers
\usepackage{graphicx} 
\usepackage{url}
\usepackage[ruled,linesnumbered,algo2e,vlined]{algorithm2e}
\usepackage{amsmath}
\usepackage[roman]{complexity}
\usepackage{todonotes}
\usepackage{wrapfig}
\usepackage{hyperref}
\usepackage{thmtools}
\usepackage{thm-restate}
\usepackage{caption}
\usepackage{cleveref}

\SetKwProg{Fn}{Function}{}{end}
\newcommand{\height}{\mathit{height}}

\newcommand{\lpf}{\mathrm{lpf}}

\newcommand{\leftmostOcc}{\mathrm{lmocc}}

\newcommand{\rlgOpt}{\hat{g}_{\mathrm{rl}}}
\newcommand{\lzhbOpt}[1]{\hat{z}_{\mathit{HB}(#1)}}
\newcommand{\lzhbOptn}[1]{\hat{\tilde{z}}_{\mathit{HB}(#1)}}

\newcommand{\LZHB}[1]{\textsf{LZHB#1}}

\newcommand{\greedyBATLZ}{\textsf{greedy-BATLZ}}
\newcommand{\greedierBATLZ}{\textsf{greedier-BATLZ}}
\newcommand{\ignore}[1]{}
\newcommand*{\myblock}[1]{\subparagraph{#1.}}

\title{Height-bounded Lempel-Ziv encodings}

\author{Hideo Bannai}{M\&D Data Science Center, Tokyo Medical and Dental University (TMDU), Japan}{hdbn.dsc@tmd.ac.jp}{https://orcid.org/0000-0002-6856-5185}{JSPS KAKENHI Grant Numbers JP20H04141, JP24K02899.}
\author{Mitsuru Funakoshi}{NTT Communication Science Laboratories, Japan}{mitsuru.funakoshi@ntt.com}{https://orcid.org/0000-0002-2547-1509}{}
\author{Diptarama Hendrian}{M\&D Data Science Center, Tokyo Medical and Dental University (TMDU), Japan}{diptarama.hendrian@tmd.ac.jp}{https://orcid.org/0000-0002-8168-7312}{}
\author{Myuji Matsuda}{Graduate School of Medical and Dental Sciences, Tokyo Medical and Dental University (TMDU), Japan}{ma190093@tmd.ac.jp}{}{}
\author{Simon J. Puglisi}{Department of Computer Science, University of Helsinki, Helsinki, Finland}{puglisi@cs.helsinki.fi}{https://orcid.org/0000-0001-7668-7636}{Academy of Finland grants 339070 and 351150.}

\authorrunning{H. Bannai, M. Funakoshi, D. Hendrian, M. Matsuda, S.\,J. Puglisi} 

\Copyright{Hideo Bannai, Mitsuru Funakoshi, Diptarama Hendrian, Myuji Matsuda, Simon J. Puglisi} 

\ccsdesc[500]{Theory of computation~Data compression}

\keywords{Lempel-Ziv parsing, data compression} 

\category{}

\relatedversion{}

\supplement{Software available at \url{https://github.com/dscalgo/lzhb/}.}
\acknowledgements{Computing resources were provided by Human Genome Center (the Univ. of Tokyo), through the usage service provided by M\&D Data Science Center, Tokyo Medical and Dental University.}

\begin{document}
\maketitle
\begin{abstract}
    We introduce
    height-bounded LZ encodings (LZHB), a new family of compressed representations that are variants of Lempel-Ziv parsings
    with a focus on bounding the worst-case access time to arbitrary positions in the text directly via the compressed representation.
    An LZ-like encoding is a partitioning of the string into phrases of length $1$ which can be encoded literally,
    or phrases of length at least $2$ which have a previous occurrence in the string and can be encoded by its position and length.
    An LZ-like encoding induces an implicit referencing forest on the set of positions of the string.
    An LZHB encoding is an LZ-like encoding where the height of the implicit referencing forest is bounded.
    An LZHB encoding with height constraint $h$ allows access to an arbitrary
    position of the underlying text using $O(h)$ predecessor queries.

    While computing the optimal (i.e., smallest) LZHB encoding efficiently seems to be difficult [Cicalese \& Ugazio 2024, arxiv],
    we give the first linear time algorithm for strings over a constant size alphabet
    that computes the greedy LZHB encoding, i.e., the string is processed from beginning to end,
    and the longest prefix of the remaining string that can satisfy the height constraint is taken as the next phrase.
    Our algorithms significantly improve both theoretically and practically,
    the very recently and independently proposed algorithms by Lipt\'ak et al. (arxiv, to appear at CPM 2024).

    We also analyze the size of height bounded LZ encodings in the context of repetitiveness measures,
    and show that there exists a constant $c$ such that the size $\lzhbOpt{c\log n}$ of the optimal
    LZHB encoding whose height is bounded by $c\log n$ for any string of length $n$
    is $O(\rlgOpt)$,
    where $\rlgOpt$ is the size of the smallest run-length grammar.
    Furthermore,
    we show that there exists a family of strings such that $\lzhbOpt{c\log n} = o(\rlgOpt)$,
    thus making $\lzhbOpt{c\log n}$ one of the smallest known repetitiveness measures for which $O(\polylog n)$ time access is possible
    using linear ($O(\lzhbOpt{c\log n})$) space.
\end{abstract}
\clearpage
\setcounter{page}{1}

\section{Introduction}\label{sec:introduction}
Dictionary compressors are a family of algorithms that produce compressed representations of input strings essentially as sequences of copy and paste operations.
These representations are widely used in general compression tools such as gzip and LZ4, and have also received much recent attention since they are especially effective for {\em highly repetitive} data sets, such as versioned documents and pangenomes~\cite{DBLP:journals/csur/Navarro21a}.

A desirable operation to support on compressed data is that of {\em random access} to arbitrary positions in the original data. Access should be supported without decompressing the data in its entirety, and ideally by decompressing little else than the sought positions. Recent years have seen intense research on the access problem in the context of dictionary compression, most notably for grammar compressors (SLPs) and Lempel-Ziv (LZ)-like schemes. The general approach is to impose structure on the output of the compressor and in doing so add small amounts of information to support fast queries.
While LZ77~\cite{DBLP:journals/tit/ZivL77}
is known to be theoretically and practically one of the smallest compressed representations
that can be computed efficiently,
a long-standing question is whether $O(\polylog(n))$ time access can be achieved with a data structure using $O(z)$ space~\cite{DBLP:journals/csur/Navarro21a}, where $z$ is the size of the LZ77 parsing.

In an LZ-like compression scheme, the input string $T$ of length $n$ is {\em parsed} into $z' \le n$ {\em phrases} (substrings of $T$),
where each phrase is of length $1$, or, is a phrase of length at least $2$ which has a previous occurrence.
Each phrase can be encoded by a pair $(\ell,s)$, where $\ell \geq 1$ is the length of the phrase,
and $s$ is the symbol representing the phrase if $\ell = 1$, or otherwise, a position of a previous occurrence (source) of the phrase.
A greedy left-to-right algorithm that takes the longest prefix of the remaining string that can be a phrase leads to the aforementioned LZ77 parsing.
An LZ-like encoding of a string induces a {\em referencing forest}, where each position of the
string is a node, phrases of length~1 are roots of a tree in the forest, and the parent of all
other positions are induced by the previous occurrence of the phrase it is contained in.
The principle hurdle to support fast access to an arbitrary symbol $T[i]$ from an LZ-like parsing is to trace the symbol to a root in the referencing forest. For LZ77, the height of the referencing forest can be $\Theta(n)$.

\myblock{Contributions} In this paper, we explore LZ-like parsings specifically designed to bound the height of the referencing forest, i.e., LZ-like parsers in which the parsing rules prevent the introduction of phrases that would exceed a specified maximum height of the referencing forest.
Our contributions are summarized as follows:
\begin{enumerate}
    \item We propose the first linear time\footnote{Assuming a constant-size alphabet.} algorithm $\textsf{LZHB}$ for computing the greedy height-bounded LZ-like encoding,
          i.e.,
          the string is processed from beginning to end, and
          each phrase is greedily taken as the longest prefix of the remaining string that can satisfy the height constraint.
          This problem was first considered by Kreft and Navarro~\cite{DBLP:conf/dcc/KreftN10},
          where an algorithm called \textsf{LZ-COST}, with no efficient implementation, is mentioned.
          Very recently, contemporaneously and independently of our work,
          this problem was also revisited by Lipt\'ak et al.~\cite{liptak2024batlz},
          who presented an algorithm called \greedyBATLZ{} running in $O(n\log^3 n)$ time.
          Lipt\'ak et al. also proposed \greedierBATLZ{}
          that runs in $O(z'n^2\log n) = O(n^3\log n)$ time, where $z'$ is the size of the parsing,
          further adding the requirement that the previous occurrence of the phrase is chosen so as to minimize the maximum height.
          We show that our algorithms can be modified to support this heuristic in
          $O(n\log \sigma + \mathit{occ}) = O(n\log \sigma + z'n)= O(n^2)$ time, where $\mathit{occ}$
          is the total number of previous occurrences of all phrases.
    \item We show that our algorithms allow for simple, lightweight implementations
          based on suffix arrays and segment trees,
          and that our implementations are an order of magnitude (or two) faster than the recent implementations of
          related schemes by Lipt\'ak et al.
    \item We propose a new LZ-like encoding, which may be of independent theoretical interest, which can be considered as a run-length variant of standard LZ-like encodings which can potentially reduce the referencing height further.
    \item We analyze the size of height bounded LZ-like encodings in the context of repetitiveness measures~\cite{DBLP:journals/csur/Navarro21a}.
          We show that for some constant $c$, there exits data structures
          of $O(\lzhbOpt{c\log n})$ size that allow access in $O(\polylog{n})$ time,
          where $\lzhbOpt{c\log n}$ is the size of the smallest LZ-like encoding whose height is bounded by $c\log n)$ .
          Furthermore $\lzhbOpt{c\log n}$ is always $O(\rlgOpt)$, and further can be $o(\rlgOpt)$ for some family of strings,
          where $\rlgOpt$ denotes the size of the smallest run-length grammar (RLSLP)~\cite{NIIBT16} representing the string.
          This makes $\lzhbOpt{c\log n}$ one of the smallest known measures that can achieve $O(\polylog n)$ time access using linear space.
          The other two are the size $z_e$ of the LZ-End parse~\cite{DBLP:conf/dcc/KreftN10,DBLP:journals/tcs/KreftN13,DBLP:conf/soda/KempaS22},
          and the size $\hat{g}_{it(d)}$ of the smallest iterated SLPs (ISLPs)~\cite{NavarroUrbinaLATIN2024}.
          Note that while there exist string families such that $\lzhbOpt{c\log n} = o(z_e)$, we do not yet know if it can be that $z_e=o(\lzhbOpt{c\log n})$.
\end{enumerate}

\section{Preliminaries}
\label{sec:prelims}
\subsection{Strings}\label{sec:prelims:strings}
Let $\Sigma$ be a set of symbols called the {\em alphabet},
and $\Sigma^*$ the set of strings over $\Sigma$.
For any non-negative integer $n$, $\Sigma^n$ is the set of strings of length $n$.
For any string $x\in\Sigma^*$, $|x|$ denotes the length of $x$.
The empty string (the string of length $0$) is denoted by $\varepsilon$.
For any integer $i\in [1,|x|]$, $x[i]$ is the $i$th symbol of $x$,
and for any integers $i,j\in[1,|x|]$,
$x[i..j] = x[i]\cdots x[j]$ if $i \leq j$, and $\varepsilon$ otherwise.
We will write $x[i..j)$ to denote $x[i..j-1]$.
For string $w$, $p\in [1,|w|]$ is a {\em period} of $w$ if
$w[i] = w[i+p]$ for all $i\in [1,|w|-p]$.
For any string $x$, $i\in [1,|w|-|x|+1]$ is an {\em occurrence} of $x$ in $w$
if $w[i..i+|x|) = x$.
For any position $i>1$, the length $\lpf_w(i)$ of the longest previously occurring factor of position $i$ is
$\lpf_w(i) = \max \{ l \mid w[i'..i'+l) = w[i..i+l), i' < i\}$, and let $\lpf_w(1)=0$.
For any $i,\ell$ with $1 \leq i \leq i+\ell \leq |w|$, the leftmost occurrence of the length-$\ell$ substring starting at $i$ is
$\leftmostOcc_w(i,\ell) = \min \{ j < i \mid w[j..j+\ell) = w[i..i+\ell) \}$, which can be undefined when $i$ is the leftmost occurrence of $w[i..i+\ell)$.
We will omit the subscript when the underlying string considered is clear.

A string that is both a prefix and suffix of a string is called a {\em border} of that string.
The {\em border array}
$B$ of a string $w$ is an array $B[1..|w|]$ of integers,
where the $B[i]$ is the length of the longest proper border of $w[1..i]$.
The border array is computable in $O(|w|)$ time and space~\cite{DBLP:journals/siamcomp/KnuthMP77},
in an on-line fashion, i.e., at each step $i = 1, \ldots, |w|$, the border array of
$B[1..i]$ is obtained in amortized constant ($O(i)$ total) time.
Notice that the minimum period of $w[1..i]$ is $i-B[i]$.
Thus, the minimum periods of all prefixes of a (possibly growing) string can be computed in time linear in the length of the string.

\begin{lemma}[\cite{DBLP:journals/siamcomp/KnuthMP77}]\label{lem:borderarray}
    The minimum period of a semi-dynamic string that allows symbols to be appended
    is non-decreasing
    and computable in amortized constant time per symbol.
\end{lemma}

The following is another useful lemma, rediscovered many times in the literature.
\begin{lemma}[e.g. Lemma 8 of~\cite{PR98}]\label{lem:internal_occ}
    The set of occurrences of a word $w$ inside a word $v$
    which is exactly twice longer than $w$ forms a single arithmetic progression.
\end{lemma}

Our parsing algorithms make use of the {\em suffix tree} data structure~\cite{Weiner73}. We assume some familiarity with suffix trees,
and give only the basic properties essential to our methods below. We refer the reader to the many textbook treatments of suffix trees for further details~\cite{gusfield_1997,Navarro-book}.

\begin{enumerate}
    \item\label{stprop:def} The suffix tree of a string $T$, denoted $\mathcal{T}_T$ (or just $\mathcal{T}$ when the context is clear) is a compacted trie containing all the suffixes of $T$.
          Each leaf of the suffix tree corresponds to a suffix of the string and is labelled with the starting position of that suffix.
    \item\label{stprop:construction}
          The suffix tree can be constructed in an online fashion via Ukkonen's algorithm~\cite{DBLP:journals/algorithmica/Ukkonen95},
          in $O(|T|\log\sigma)$ total time,
          i.e., at each step $i=1,\ldots, |T|$, the
          suffix tree of $T[1..i]$ is obtained in amortized $O(\log\sigma)$ ($O(i\log\sigma)$ total) time.
          For linearly-sortable alphabets, it is possible to compute the suffix tree offline in $O(|T|)$ time~\cite{DBLP:conf/focs/Farach97}.
    \item\label{stprop:prefixquery}
          Suffix trees can support {\em prefix queries} that return a pair $(\ell, s)$, where $\ell$ is the length of the longest prefix of the query string that occurs in the indexed string, and $s$ is the leftmost position of an occurrence of that prefix.
          Prefix queries can be conducted simply by traversing the suffix tree from the root with the query string
          and take $O(\ell\log \sigma)$ time.
          The query string can be processed left-to-right, where each symbol is processed in $O(\log\sigma)$ time, which is the cost of finding the correct out-going child edge
          from the current suffix tree node.
          Prefix queries can be answered in the same time complexity even if the underlying string $T$ is extended in an online fashion as mentioned in Property~\ref{stprop:construction}.
    \item\label{stprop:leftmostocc_of_factor}
          The suffix tree of $T$ can be preprocessed in linear time to answer, given $i,\ell$, the leftmost occurrence in $T$ of the substring $T[i..i+\ell)$,
          i.e.,
          $\leftmostOcc_T(i,\ell)$
          in constant time~\cite{DBLP:conf/cpm/BelazzouguiKPR21}.
\end{enumerate}

\subsection{LZ encodings and random access}
An LZ-like parsing of a string is a decomposition of the string
into phrases of length $1$ (literal phrases),
or a phrase of length at least $2$ which has a previous occurrence in the string.
An LZ-like {\em encoding} is a representation of the LZ-like parsing,
where literal phrases are encoded as the pair
$(1,c)$, where $c\in\Sigma$ is the phrase itself, and phrases of length at least $2$ are
encoded as the pair $(\ell,s)$, where $\ell \geq 2$ is the length of the phrase and
$s$ is a previous occurrence (or the {\em source}) of the phrase.
Although LZ parsing and encoding are sometimes used as synonyms,
we differentiate them in that
a parsing only specifies the length of each phrase, while
an encoding specifies the previous occurrence of each phrase as well.
The {\em size} of an LZ-like encoding is the number of phrases.
For example, for the string
$\mathtt{ababacbabac}$,
we can have an encoding
$(1,\mathtt{a}),(1,\mathtt{b}),(3,1),(1,\mathtt{c}),(5,2)$
of size $5$.
A common variant of LZ-like encodings adds an extra symbol explicitly to the phrase,
and each phrase is encoded by a triplet.
For simplicity, our description will not consider this extra symbol;
the required modifications for the algorithms to include this are straightforward.
We note that our experiments in Section~\ref{sec:experiments} will include the extra character,
in order to compare with implementations of previous work.

\begin{wrapfigure}[11]{r}{0.5\textwidth}
    \centering
    \includegraphics[width=0.5\textwidth]{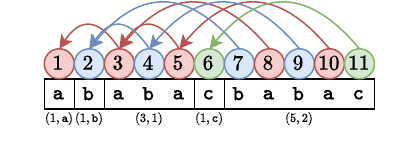}
    \caption{An example of an implicit referencing forest induced from the LZ-like encoding $(1,\mathtt{a}),(1,\mathtt{b}),(3,1),(1,\mathtt{c}),(5,2)$ for the string $\mathtt{ababacbabac}$.}
    \label{fig:example:ref_forest}
\end{wrapfigure}
An LZ-like encoding of a string induces an implicit referencing forest, where:
each position of the string is a node, literal phrases are roots of a tree in the forest,
and the parent of all other positions are induced by the source of the phrase it is contained in.
For example, for the encoding of the string $\mathtt{ababacbabac}$ as above,
the referencing forest is shown in Figure~\ref{fig:example:ref_forest}.
Let $\mathcal{E} = (\ell_1, s_1), \ldots, (\ell_{z'}, s_{z'})$ be an LZ-like encoding of string $T$.
Given an arbitrary position~$i$, suppose we would like to retrieve the symbol $T[i]$.
This can be done by traversing the implicit referencing forest using predecessor queries.
We first find the phrase $(\ell_j, s_j)$ with starting position $b_j$\footnote{We note that it is possible to encode each phrase as $(b_j,s_j)$, since then, $\ell_j = b_{j+1}-b_j$.}
that the position $i$ is contained in, i.e.,
$j$ s.t. $b_j \leq i < b_j + \ell_j$.
Then, we can deduce that the parent position of $i$ is $i' = s_j + (i - b_j)$.
This is repeated until a literal phrase $(1, c)$ is reached, in which case $T[i] = c$.
The number of times a parent must be traversed (i.e., the number of predecessor queries) is bounded by the
    {\em height} of the referencing forest.
This notion of height for LZ-like encodings was introduced by Kreft and Navarro~\cite{DBLP:journals/tcs/KreftN13}.

\section{Height bounded LZ-like encodings}
\label{sec:lzhb}
We first consider modifying the definition of the height of the referencing forest under some conditions, by implicitly rerouting edges.
Consider the case of a unary string $T=\mathtt{a}^n$ and LZ-like encoding $(1,\mathtt{a}), (n-1,1)$. The straightforward definition above gives a height of $n-1$ for position $n$.
This is a result of the second phrase being self-referencing,
i.e., the phrase overlaps with its referenced occurrence,
and each position references its preceding position.
An important observation is that self-referencing phrases are periodic; in particular a self-referencing phrase starting at position $b_i$ with source $s_i$ has period $b_i-s_i$.
Due to this periodicity, any position in the phrase can refer to an appropriate position in $T[s_i..b_i)$.

More formally, let $\mathcal{E} = (\ell_1, s_1), \ldots, (\ell_{z'}, s_{z'})$ be an LZ-like encoding
of $T$. For any position $i$,
let $b_j$ be the starting position of the phrase $(\ell_j,s_j)$
such that $b_j \leq i < b_j+\ell_j$.
Then, $T[i] = s_j$ if $\ell_j = 0$, and $T[i] = T[s_j + (i - b_j)] = T[s_j+((i-b_j)\bmod (b_j-s_j))]$ otherwise.

We define the height $\height_\mathcal{E}(i)$ of position $i$ by an encoding $\mathcal{E}$ as
\begin{equation}
    \height_\mathcal{E}(i) = \begin{cases}
        0                                                    & \mbox{if $\ell_j = 0$} \\
        \height_\mathcal{E}(s_j+((i-b_j)\bmod (b_j-s_j))) +1 & \mbox{otherwise.}
    \end{cases}\label{eq:height:period}
\end{equation}
The subscript will be omitted if the underlying encoding considered is clear.
Since $j$ can be computed using a predecessor query on the set of phrase starting positions,
$T[i]$ can be computed using $O(\height(i)Q(z))$ time using a data structure of size $O(z)$,
where $Q(z)$ is the time required for predecessor queries on $z$ elements in $[1,n]$ using $O(z)$ space.
$Q(z)$ is $O(\log n)$ using a simple binary search, and faster using more sophisticated methods~\cite{NRacmcs20}.

For example, for the encoding
$(1,\mathtt{a}),(1,\mathtt{a}),(1,\mathtt{b}),(3,2),(1,\mathtt{c}),(4,3)$
of string $\mathtt{aababacbaba}$,
the heights are:
$0,0,0,1,1,1,0,1,2,2,2$.
Notice that the phrase $(3,2)$ that starts at position $4$, representing $\mathtt{aba}$ is self-referencing, and the height of the last position in the phrase (position 6) is $1$,
since it is defined to reference position
$2 = 2 + (6-4)\bmod(6-4)$.
We note that even with this modified definition, the height of the optimal LZ-like encoding (LZ77) can be $\Theta(n)$.
See \cref{observation:linear-height} in~\cref{appendix:claims}.

In order to bound the worstcase query time complexity for access operations, we consider LZ-like encodings with bounded height.
An {\em $h$-bounded LZ-like encoding} is an LZ-like encoding where $\max \{ \height(i) \mid i\in [1,n]\} \leq h$.

There are many ways one could enforce such a height restriction. Unfortunately,
finding the smallest such encoding was very recently shown to be NP-hard by Cicalese and Ugazio~\cite{cicalese2024complexity}.
Therefore, we propose several greedy heuristics to compute $h$-bounded LZ-like encodings.
Given an encoding for $T[1..b_j)$ (which defines $\height(i)$ for any $i\in [1,b_j)$),
the next phrase $(\ell_j, s_j)$ starting at position $b_j$ is defined as follows.
\begin{description}
    \item [\LZHB{1}]
          $\ell_j$ is the largest value such that
          $T[s_j..s_j+\ell_j)$ satisfies the height constraint,
          where $s_j = \leftmostOcc_T(b_j,\lpf_T(b_j))$, i.e.,
          the leftmost occurrence of the longest previously occurring factor at position $b_j$.

    \item [\LZHB{2}]
          $\ell_j$ is the largest value such that for all $1 \leq \ell' \leq \ell_j$, the leftmost occurrence of $T[b_j..b_j+\ell')$ satisfies the height constraint. $s_j$ is the leftmost occurrence of $T[b_j..b_j+\ell_j)$.

    \item [\LZHB{3}]
          $\ell_j$ is the largest value such that for all $1 \leq \ell' \leq \ell_j$, there exists some previous occurrence of $T[b_j..b_j+\ell')$ that satisfies the height constraint.
          $s_j$ is the leftmost occurrence of $T[b_j..b_j+\ell_j)$ that satisfies the height constraint.
\end{description}
Note that for all variations, $\ell_j=1$ if
there is no occurrence of $T[b_j]$ that satisfies the corresponding conditions described above, in which case, $s_j = T[b_j]$.

We note that
\LZHB{1} corresponds to the baseline algorithm \textsf{BATLZ2} in~\cite{liptak2024batlz}.
\LZHB{3} essentially corresponds to \greedyBATLZ{} in~\cite{liptak2024batlz} as well,
but \greedyBATLZ{} does not require that the occurrence is leftmost.
The difference between \LZHB{2} and \LZHB{3} lies in the priority of the choice of the leftmost occurrence and when to check the height constraint.
\LZHB{2} greedily extends the prefix, checking each time whether its leftmost occurrence satisfies the height constraint.
\LZHB{3} finds the leftmost occurrence out of the longest prefix that can satisfy the height constraint.
Note that when the height is unbounded,
the size of all three variants will be equivalent to the regular LZ77 parsing.

We show that \LZHB{1} and \LZHB{2} can be computed in $O(n)$ time and space for linearly-sortable alphabets,
and \LZHB{3} can be computed online
in $O(n\log\sigma)$ time and $O(n)$ space for general ordered alphabets.

\medskip

We also propose a new encoding of LZ-like parsings that re-routes edges of the implicit referencing forest in an attempt to further reduce the heights, again using periodicity.
In the case of self-referencing phrases, our definition of height in Equation~(\ref{eq:height:period})
utilized the fact that self-referencing phrases implies a period $b_j-s_j$ in the phrase.
Using this period, we could re-route the parents of all positions inside the phrase to the first period in the previous occurrence of the phrase,
which is outside the phrase.
Here, we make two further observations:
1) since the referenced substring is the first period of the phrase,
we could extended the phrase for free while the period continues,
possibly making the phrase longer,
and
2) the referenced substring could have a previous occurrence further to the left,
which should tend to have shorter heights.
Therefore, if we were to store the period of the phrase explicitly, this could be applied to and benefit all periodic phrases, self-referencing or not.

Specifically then, under this scheme, a phrase with period $1$ can be encoded as the triple $(\ell,c,1)$, where $c$ is a symbol, and a phrase with period $p \geq 2$ can be encoded as the triple $(\ell, s, p)$,
where $\ell \geq p$ is the length of the phrase with period $p \geq 2$, and $s$ is a previous occurrence of the length-$p$ prefix of the phrase.
We will call such an encoding, a {\em modified}
LZ-like encoding.
For example,
$(2,\mathtt{a},1),(1,\mathtt{b},1),
    (3,2,2),(1,\mathtt{c},1),(4,3,2)$
would be a modified LZ-like encoding
for the string $\mathtt{aababacbaba}$.

The implicit referencing forest and heights of the modified LZ-like encoding can be defined analogously:
a position $i$ in the $j$th phrase $(\ell_j,s_j,p_j)$ that begins at position $b_j$, i.e., $b_j \leq i < b_j+\ell_j$,
will reference position $s_j+((i-b_j)\bmod p_j)\bmod (b_j - s_j)$, where the second $\bmod$ is to deal with the
case where the occurrence of the prefix period is self-referencing.
For the above example, the heights are:
$0,0,0,1,1,1,0,1,2,1,2$.

A greedy left-to-right algorithm computes the $j$th phrase
$(\ell_j, s_j, p_j)$ starting at $b_j$ as:
\begin{description}
    \item [\LZHB{4}]
          $\ell_j$ is the largest value such that
          $T[b_j..b_j+\ell_j)$ has period $1$, or,
          for all $1 \leq \ell' \leq \ell_j$,
          there exists some previous occurrence of $T[b_j..b_j+p')$  that satisfies the height constraint, where $p'$ is the minimum period of $T[b_j..b_j+\ell')$. $s_j$ is the leftmost occurrence of $T[b_j..b_j+\ell_j)$ that satisfies the height constraint, and $p_j$ is the minimum period of $T[b_j..b_j+\ell_j)$.
\end{description}

\section{Efficient parsing algorithms for height-bounded encodings}
\label{sec:parsing}
\subsection{A linear time algorithm for \LZHB{1}}
\begin{theorem}\label{thm:lzhb1}
    For any integer $h$, an $h$-bounded encoding based on \LZHB{1} for a string over a linearly-sortable alphabet can be computed in linear time and space.
\end{theorem}
\begin{proof}
    The algorithm maintains an array $H[1..n]$ of integers, initially set to 0.
    Phrases are produced left to right.
    The algorithm maintains the invariant that when the encoding is computed up to position $i$, $H[j]$ gives the height in the referencing forest of position $j<i$.
    Using Property~\ref{stprop:construction} in Section~\ref{sec:prelims:strings},
    we build the suffix tree of $T$ in linear time, and preprocess it, again in linear time, for Property~\ref{stprop:leftmostocc_of_factor}.
    We also precompute and store all
    the lengths of the longest previously occurring factor at each position, i.e., $\lpf_T(i)$ for all $1 \leq i \leq n$, in linear time~\cite{DBLP:journals/ejc/CrochemoreIIKRW13}.

    Suppose we have computed the $h$-bounded encoding for $T[1..b_j)$ and would like to compute the $j$th phrase $(\ell_j,s_j)$ starting at position $b_j$.
    We can compute
    $s_j = \leftmostOcc_T(b_j,\lpf_T(b_j))$, i.e.,
    the leftmost occurrence of the longest previously occurring factor starting at position $b_j$,
    in constant time (Property~\ref{stprop:leftmostocc_of_factor}).
    The encoding of the phrase starting at $b_j$ is then $(s_j,\ell_j)$, where $\ell_j \leq \ell$ is the largest value such that every value in $H[s_j..\min(b_j,s_j+\ell_j))$ is less than $h$.
    This can be computed in $O(\ell_j)$ time by simply scanning the above values in $H$.
    Notice that we do not need to check the height beyond position $b_j$, since this would imply that the phrase is self-referencing, and the remaining heights will be copies of the first period of the phrase which
    were already checked to satisfy the height constraint.
    After pair $(s_j,\ell_j)$ is determined,
    values in $H[b_j..b_j+\ell_j)$ are determined according to Equation~(\ref{eq:height:period}), in $O(\ell_j)$ time.

    The total time is thus proportional to the sum of the phrase lengths, which is $O(n)$.
\end{proof}

\subsection{A linear time algorithm for \LZHB{2}}\label{sec:lzhb2}
The difference between \LZHB{1} and \LZHB{2} is that while \LZHB{1} fixes the source to the leftmost occurrence of the longest previously occurring factor starting at $b_j$,
\LZHB{2} considers the leftmost position of the candidate substring.
Since shorter lengths allow the source to be further to the left where heights tend to be shorter, it may allow longer phrases.

We show below that we can still achieve a linear time parsing algorithm.
\begin{theorem}\label{thm:lzhb2}
    For any integer $h$, an $h$-bounded encoding based on \LZHB{2} for a string over a linearly-sortable alphabet can be computed in linear time and space.
\end{theorem}
\begin{proof}
    The same steps as in the first paragraph in the description of $\LZHB{1}$ in \cref{thm:lzhb1} are taken
    (except for the precomputing of $\lpf_T$ which is not needed here).

    Suppose we have computed the $h$-bounded encoding for $T[1..b_j)$ and would like to compute the $j$th phrase $(\ell_j,s_j)$ starting at position $b_j$.
    We start with $\ell=1$, and increase $\ell$ incrementally until there is no previous occurrence of $T[b_j..b_j+\ell)$, or the leftmost occurrence of $T[b_j..b_j+\ell)$ violates the height constraint, and will use the last valid value of $\ell$ for $\ell_j$.

    For a given $\ell$, using Property~\ref{stprop:leftmostocc_of_factor},
    we can obtain the leftmost occurrence of $T[b_j..b_j+\ell)$, i.e., $s'_j = \leftmostOcc_T(b_j,\ell)$, in constant time.
    We need to check whether
    all values in
    $H[s'_j..\min(b_j,s'_j+\ell))$
    are less than $h$.
    Since $s'_j$ may change for different values of $\ell$,
    we potentially need to check all values in $H[s'_j..\min(b_j,s'_j+\ell))$.

    To do this efficiently,
    we maintain another array $C$, where $C[i]$ stores the position $j$ in $H$,
    such that $j$ is the smallest position in $H$ that is greater than or equal to $i$,
    such that $H[j] = h$, if it exists, or $\infty$ otherwise, i.e.,
    $C[i] = \min (\{\infty\} \cup \{ j \geq i \mid H[j] = h \})$.
    Then, $\max(H[i..j]) +1 \leq h$ if and only if $C[i] > j$. Thus, checking the height constraint can be done in constant time for a given $\ell$.

    After $(s_j,\ell_j)$ are determined,
    values in $H[i..i+\ell_j)$ are determined according to Equation~(\ref{eq:height:period}), in $O(\ell_j)$ time.
    The array $C$ can be maintained in linear total time: all values in $C$ are initially set to $\infty$,
    and whenever we store the value $h$ in $H[j]$ for some $j$, we set $C[i']= j$ for all
    $i'$ such that $\max (\{ 0\} \cup \{ i < j \mid H[i] = h \}) < i' \leq j$,
    i.e., we update all values in $C[1..i)$ that were $\infty$.
    Clearly
    each element of $C$ is set at most once.
    Thus, the total time is $O(n)$.
\end{proof}

\subsection{An $O(n\log \sigma)$ time algorithm for \LZHB{3}}\label{sec:lzhb3}

\begin{theorem}
    For any integer $h$, an $h$-bounded encoding based on \LZHB{3} for a string  can be computed in $O(n\log \sigma)$ time and $O(n)$ space.
\end{theorem}
\begin{proof}
    Our algorithm is a slight modification of a folklore algorithm to compute LZ77 in $O(n\log\sigma)$ time based on Ukkonen's online algorithm for computing suffix trees~\cite{DBLP:journals/algorithmica/Ukkonen95} described, e.g., by Gusfield~\cite{gusfield_1997}.
    We first describe a simple version which does not allow self-references.

    At a high level, when computing the $j$th phrase that starts at position $b_j$, we use the suffix tree of $T[1..b_j)$ so that we can find the longest previous occurrence of $T[b_j..]$. We can find this by simply traversing the suffix tree from the root with $T[b_j..]$ as long as possible (Property~\ref{stprop:prefixquery} in Section~\ref{sec:prelims:strings}).
    Once we know the length $\ell_j$ of the phrase, we simply append the phrase $T[b_j..b_j+\ell_j)$ to the suffix tree to obtain the suffix tree for $T[1..b_{j+1})$, and continue.
    An occurrence
    (particularly, the leftmost occurrence)
    can be obtained easily by storing this information on each edge when it is constructed.

    For our height-bounded case, the precise height of a position in a phrase can only be determined after $\ell_j$ is determined, and if $\ell_j \geq 2$,
    the leftmost occurrence $s_j$ of $T[b_j..b_j+\ell)$ that satisfies the height constraint must be determined.
    Since a reference adds 1 to the height, we know that the height constraint is satisfied if and only if we reference positions that have height less than $h$.
    Therefore, when adding the symbols of the phrase $T[b_j..b_j+\ell_j)$ to the suffix tree
    after determining their heights,
    we change symbols corresponding to positions with height $h$ to a special symbol $\$$ which does not match any other symbols.
    In other words, for $T[1..b_j)$ and their heights $H[1..b_j)$,
    we use and maintain a suffix tree for $T'[1..b_j)$,
    where $T'[i] = T[i]$ if $H[i] < h$ and $T'[i] = \$$ otherwise.
    This simple modification allows us to
    reference only and all substrings with heights less than $h$,
    so that we can find the longest prefix of $T[b_j..]$ that occurs in $T[1..b_j)$ and also satisfies the height constraint, by simply traversing the suffix tree of $T'[1..b_j)$ (Property~\ref{stprop:prefixquery}).

    Next, in order to allow self-references, we utilize a combinatorial observation made previously, that self-references lead to periodicity.
    Let $x = T[b_j..b_j+|x|)$ be the longest prefix of
    $T[b_j..n]$ such that there is a non-self-referencing occurrence
    in $T'[1..b_j)$.
    We can find a longer self-referencing prefix of $T[b_j..n]$, if it exists, as follows:
    Any such occurrence must be prefixed by $x$,
    and in order for the phrase to be self-referencing, this $x$
    must occur as a substring in $T[b_j-|x|..b_j+|x|-1)$.
    Since the heights after the first period of the phrase are essentially defined as copies of the heights of the first period of the phrase (see Equation~(\ref{eq:height:period})),
    any occurrence of $x$ (and any of its extensions) at or after the leftmost position
    $t_j$ in $[b_j-|x|,b_j)$ such that $H[t_j..b_j) < h$, would satisfy the height constraint.

    If there is no occurrence of $x$ in
    $T[\max(t_j,b_j-|x|)..b_j+|x|-1)$,
    then a longer self-referencing phrase satisfying the height constraint cannot exist, so $\ell_j=|x|$.
    If there is only one such occurrence of $x$, we can use this occurrence and extend the phrase by naive symbol comparisons to obtain $\ell_j$.
    If there are more than one such occurrence,
    it follows that there are at least three occurrences of $x$ in $T[\max(t_j,b_j-|x|)..b_j+|x|)$,
    a string of length (at most) $2|x|$.
    From Lemma~\ref{lem:internal_occ}, this implies that the occurrences of $x$ form an arithmetic progression,
    whose common difference is the minimum period of $x$. Thus, extending the phrase as above from any occurrence of $x$ in $T[\max(t_j,b_j-|x|)..b_j+|x|-1)$ will lead to the same maximal length due to the periodicity, i.e., $\ell_j$ is the maximum value such that $T[b_j..b_j+\ell_j)$ has the same minimum period as $x$.
    We can find all such occurrences
    (which obviously includes the leftmost)
    in $O(\ell_j)$ time using any linear time pattern matching algorithm, e.g., KMP~\cite{DBLP:journals/siamcomp/KnuthMP77}.

    The time to compute phrase $(\ell_i,s_i)$ is $O(\ell_i\log\sigma)$, so we spend $O(n\log\sigma)$ time in total.
\end{proof}
Pseudo-code for the algorithm is shown in Algorithm~\ref{alg:lz77h3_selfreferencing2} in \cref{appendix:pseudocode}.

\subsection{An $O(n\log \sigma)$ time algorithm for \LZHB{4}}
\begin{theorem}
    For any integer $h$, an $h$-bounded encoding for a string based on \LZHB{4} can be computed in $O(n\log \sigma)$ time and $O(n)$ space.
\end{theorem}
\begin{proof}
    The algorithm is similar to \LZHB{3}, in that it computes the suffix tree of $T[1..b_j)$ in an online manner, and after determining the new phrase $T[b_j..b_j+\ell_j)$ and their heights $H[b_j..b_j+\ell_j)$,
    symbols of the new phrase are added to the suffix tree
    except when its position has height $h$, in which case $\$$ is added.

    In order to determine the \LZHB{4} phrase, we first compute the \LZHB{3} phrase, i.e.,
    the longest prefix of $T[b_j..n]$ that has a previous occurrence satisfying the height constraint. Let its length be $\ell'$.
    This implies that any prefix period of $T[b_j..n]$ that can be referenced (and satisfies the
    height constraint) is at most $\ell'$.
    Since prefix periods are non-decreasing (Lemma~\ref{lem:borderarray}),
    we have that the \LZHB{4} phrase $T[b_j..b_j+\ell_j)$ is the longest prefix of $T[b_j..n]$ with period $p_j \leq \ell'$.
    This can be computed in $O(\ell_j)$ time (again Lemma~\ref{lem:borderarray}).
    The leftmost occurrence of $T[b_j..b_j+p_j)$ to be referenced
    can be found simply by traversing the suffix tree if $p_j$ is at most the length of the
    longest non-self-referencing prefix of the \LZHB{3} phrase,
    and otherwise, will coincide with (the starting position of) the self-referencing occurrence of the \LZHB{3} phrase.
    Since $\ell' \leq \ell_j$ must hold as well, each \LZHB{4} phrase is computed in $O(\ell_j\log\sigma)$ time, and thus, \LZHB{4} can be computed in $O(n\log\sigma)$ total time and $O(n)$ space.
\end{proof}

\subsection{Algorithms for the greedier heuristic}
Lipt\'ak et al.~\cite{liptak2024batlz} also propose a {\em greedier} heuristic
for height bounded LZ-like encodings,
in which, for any phrase, the previous occurrence that minimimzes the maximum height is chosen.
While a na\"ive algorithm runs in $\Theta(n^2)$ time, Lipt\'ak et al. propose an algorithm
for which they show an upper bound of $O(z'n^2\log n) = O(n^3\log n)$ time.
We show here that \LZHB{3} and \LZHB{4} can be modified to support the greedier heuristic in total
$O(n\log\sigma + \mathit{occ})$ time, where $\mathit{occ}=O(z'n)$ is the total number of previous occurrences of all phrases.

Suppose we are able to obtain all previous occurrences of the phrases.
The array $H$ is not static, but only append operations are performed on it, and so a range maximum query data structure for $H$ can be maintained
in amortized $O(1)$ time per update and query~\cite{DBLP:conf/wads/Fischer11,DBLP:conf/sofsem/UekiDKMNYBIS17}.
Using this, the maximum height of a given occurrence can be checked in amortized $O(1)$ time, and the occurrence giving the smallest maximum height can be found in additional $O(n+\mathit{occ})$ time.

All previous occrrences of a phrase can be obtained in
additional $O(n+\mathit{occ})$ total time using standard techniques on the suffix tree.
Recall that when computing the $j$th phrase, our algorithm traverses the suffix tree for finding a longest non-self-referencing occurrence.
The leaves below the reached position contain the occurrences of the phrase,
and since any non-leaf node of a suffix tree has at least two children,
these leaves can be found in time linear in their number,
thus in additional $O(\mathit{occ})$ total time.
All self-referencing occurrences can be found in $O(\ell_j)$ time as described in Section~\ref{sec:lzhb3}.
A~minor detail we have skipped is that since we only have the implicit suffix tree being built
via Ukkonen's online construction algorithm,
the suffix tree does not have leaves corresponding to suffixes that have a previous occurrence.
This can be dealt with as follows.
Since Ukkonen's algorithm maintains the longest repeating suffix of the current text as the
    {\em active point} in the suffix tree, we can also maintain some previous occurrence $T[u..v]$ of it.
Since any previous occurrence of the phrase contained in the longest repeating suffix can be mapped to an occurrence in $T[u..v]$,
we can, given all the occurrences of the phrase in $T[u..v]$ for which there is a corresponding leaf in the suffix tree,
find all occurrences in the longest repeating suffix in
$O(n+\mathit{occ})$ total time.

\subsection{Implementation using suffix arrays}
The {\em suffix array}~\cite{DBLP:conf/soda/ManberM90} of a string $T$ of length $n$,
is an array $\mathit{SA}[1..n]$ of integers such that $T[\mathit{SA}[i]..n]$
is the $i$th lexicographically smallest suffix of $T$.

Suffix arrays are well known as a lightweight alternative to suffix trees.
While it is not difficult to use suffix arrays in place of suffix trees for static strings~\cite{DBLP:journals/jda/AbouelhodaKO04}
by mapping nodes of the suffix tree to ranges in the suffix array,
it is not straightforward to do this for \LZHB{3} and \LZHB{4}, since the algorithms work in an online manner:
the string $T'$ for which the suffix tree is maintained is determined during the computation,
and is not known in advance.
Here, we show how to simulate the algorithm on the suffix tree for $T'$
by using the suffix array $\mathit{SA}$ of $T$,
in amortized $O(\log n)$ time per suffix tree operation,
thus obtaining an $O(n\log n)$ time algorithm for \LZHB{3} and \LZHB{4}.
The running times of the greedier versions can be similarly bounded by
$O(n\log n + \mathit{occ}\log n) = O(n\log n + z'n\log n) = O(n^2\log n)$.

In a suffix tree of $T$, the effect of replacing a symbol $T[i]$ at position $i$ with $\$$
can be viewed as an operation that truncates the paths of suffixes starting at position
$j < i$ for which $T[j..i)$ does not contain $\$$, to length $i-j$, i.e., to $T[j..i)$.
We will maintain an array $L$ of integers, where $L[\mathit{rank}[i]]$ holds
the valid length of the suffix starting at position $i$ and $\mathit{rank}[i]$ is the lexicographic rank of the suffix $T[i..n]$,
i.e., $\mathit{SA}[\mathit{rank}[i]] = i$.
All values are initially set to $0$, indicating that the suffix has not yet been inserted into the suffix tree.
During the online construction of the suffix tree of $T'$,
when the suffix $T[i..n]$ is added to the suffix tree,
we set $L[\mathit{rank}[i]] = n-i+1$.
When a $\$$ symbol is appended, we modify the values in~$L$ of relevant suffixes.
Since the value of each position in $L$ is modified at most twice (from $0$ to $n-i+1$ and then to the length up to the next $\$$),
the total number of updates on $L$ is at most $2n$.

The traversal on the suffix tree can then be simulated by a standard search for finding the lexicographic range in the suffix array
of suffixes that are prefixed by the considered substring. This range can be computed in $O(\log n)$ time per symbol
by a simple binary search.
In order to make sure we do not traverse truncated suffixes, we allow the traversal if and only if there
exists at least one suffix in the lexicographic range whose valid length is at least the length of the substring being traversed.
This can be checked in $O(\log n)$ time by using a segment tree~\cite{DBLP:journals/tc/BentleyW80,DBLP:journals/siamcomp/Chazelle88} for representing $L$,
which allows updates and range maximum queries on $L$ in $O(\log n)$ time.
All occurrences, i.e., suffixes in the range for which the value in $L$
is at least the length of the substring,
can be enumerated in $O(\log n)$ time per occurrence,
by a common technique~\cite{DBLP:conf/soda/Muthukrishnan02}
that calls range maximum queries recursively on sub-ranges excluding the maximum value,
until all ranges only contain values less than the desired length.

A minor difference with the suffix tree version is that because we are able to insert the whole suffix starting at a given position rather than just a single symbol at the position,
the algorithm will naturally handle self-referencing occurrences, and they no longer require special care.
Also, while retrieving an occurrence is easy, we note that we can no longer retrieve the {\em leftmost} occurrence for the greedy versions in the same time bound - the chosen occurrence will be a prefix of the suffix with longest valid length.

\section{Height bounded LZ-like encodings as repetitiveness measures}
\label{sec:optimal}
Here, we consider the strength of height bounded encodings as repetitiveness measures~\cite{DBLP:journals/csur/Navarro21a}.
Denote by $\lzhbOpt{h}$, the optimal (smallest) LZHB encoding whose height is at most $h$.
It easily follows that $\lzhbOpt{h}$ is monotonically non-increasing in $h$, i.e., for any $h \leq h'$, it holds that $\lzhbOpt{h} \geq \lzhbOpt{h'}$.
We will also denote by $\lzhbOptn{h}$,
the optimal modified LZ-like encoding whose height is at most $h$.

A first, and obvious, observation is that if the height is allowed to grow to $n$, then the size of the height-bounded encoding and the LZ parsing are equivalent.
\begin{observation}\label{obs:nobound_is_lz77}
    $\lzhbOpt{n} = z$.
\end{observation}

It is also interesting to note height-bounded and run-length encodings can be related:
\begin{observation}
    $\lzhbOptn{0}$ is equivalent to the run-length encoding.
\end{observation}

\medskip

The following relation between $\lzhbOpt{h}$ and
the smallest size $\rlgOpt$ of RLSLPs (SLPs with run-length rules)  and $\hat{g}_{it(d)}$ of ISLPs~\cite{NavarroUrbinaLATIN2024} can be shown.

\begin{restatable}{theorem}{lowerboundGrl}\label{thm:lowerboundGrl}
    There exists a constant $c$ such that $\lzhbOpt{c\log n} = O(\rlgOpt)$ holds.
\end{restatable}

\begin{restatable}{theorem}{lowerboundGrlAndGit}\label{thm:lowerbound_grl_and_git}
    There exist a family of strings such that for some constant $c$, $\lzhbOpt{c\log n} = o(\rlgOpt)$ and $\lzhbOpt{c\log n} = o(\hat{g}_{it(d)})$.
\end{restatable}

When there is no height constraint,
\LZHB{3} is equivalent to LZ77 and thus gives an optimal LZ-like parsing.
For modified LZ-like encodings with no height contraints,
\LZHB{4} can be worse than optimal.
For example, for the string $\mathtt{abaxabcdababca}$, greedy gives
$\mathtt{a}|\mathtt{b}|\mathtt{a}|\mathtt{x}|\mathtt{ab}|\mathtt{c}|\mathtt{d}|\mathtt{abab}|\mathtt{c}|\mathtt{a}$ while
$\mathtt{a}|\mathtt{b}|\mathtt{a}|\mathtt{x}|\mathtt{ab}|\mathtt{c}|\mathtt{d}|\mathtt{ab}|\mathtt{abca}$
is a slightly smaller parsing.
However, we can show that the greedy algorithm has an approximation ratio of at most $2$.
\begin{restatable}{theorem}{twoApproximation}
    Let $\tilde{z}$ denote the size of the modified LZ-like encoding generated by \LZHB{4} with no height constraints.
    We also denote $\lzhbOptn{n}=\hat{\tilde{z}}$.
    Then, $\hat{\tilde{z}}
        \leq \tilde{z} \leq z \leq
        2\hat{\tilde{z}}$.
\end{restatable}

\section{Computational Experiments}
\label{sec:experiments}
We have developed prototype implementations of our algorithms described above, in C++.
Since Lipt\'ak et al.~\cite{liptak2024batlz} have already demonstrated the effectiveness and superiority
of \greedyBATLZ{} and \greedierBATLZ{} (which correspond to \LZHB{3} and the greedier \LZHB{3})
compared to the baselines and that the height can be reduced without increasing the size of the parse too much,
our experiments here focus
on the performance of our algorithms \LZHB{3} and \LZHB{4} in comparison with
\greedyBATLZ{} and \greedierBATLZ{}.
For the \textsf{BAT-LZ} variants, we use the code provided by Lipt{'a}k et al.\footnote{Available at:~\url{https://github.com/fmasillo/BAT-LZ}.} with slight modifications
to fit our test environment. Our implementations for \LZHB{3} and \LZHB{4} are available at~\url{https://github.com/dscalgo/lzhb/}.

Experiments were conducted on a system with
dual AMD EPYC 9654 2.4GHz processors and 768GB RAM running RedHat Linux 8.
We use the {\em Real} subset of the repetitive corpus\footnote{\url{https://pizzachili.dcc.uchile.cl/repcorpus.html}}.

Figure~\ref{figure:BATLZ-time-comparison} shows the running times for various height constraints
of the \textsf{BAT-LZ} variants as well as
\LZHB{3}, \LZHB{3SA}, \LZHB{4}, and \LZHB{4SA},
and their greedier versions,
where the suffix SA denotes the suffix array implementation.
Although \greedyBATLZ{} is fast when there is no height constraint,
its performance seems to diminish rapidly when a height constraint is set.
Interestingly, \greedierBATLZ{} seemed to run faster than \greedyBATLZ{}.
We observe that all versions of our \LZHB{3} and \LZHB{4} implementation outperform
\greedyBATLZ{} and \greedierBATLZ{} by a large margin.
For the greedier variants, the running times increase as the height constraint decreases.
This is because stricter height constraints generally lead to
shorter phrase lengths and therefore increases $z'$ and $\mathit{occ}$.
This becomes more apparent in the suffix array versions,
perhaps because the $n$ in the $O(\log n)$ factor for finding the suffix array range is actually the size of the range considered, which becomes larger for shorter pharses.

Figure~\ref{figure:BATLZ-space-comparison} shows the memory usage of the above mentioned implementations.
The memory usage for \LZHB{3} and \LZHB{4} (the suffix tree versions) are smallest for small heights,
because smaller height constraints imply that many paths in the suffix tree become truncated.
On the other hand, the memory usage is the highest for larger heights.
Memory usage for the other versions is unaffected by the height constraint.
\LZHB{3SA} and \LZHB{4SA} are always superior to \greedyBATLZ{},
and the same can be said for the greedier versions and their counterparts.

In summary, our suffix array implementations are always faster and use less memory than their {\sf BATLZ} counterparts,
with the exception of \greedyBATLZ{} with no height constraints.

Figure~\ref{figure:BATLZ-size-comparison} shows the sizes of the phrases of \LZHB{3}, \LZHB{4} and their greedier versions\footnote{We also ran experiments for LZ End parsing, but do not include results in the figures because they distort the plots. We include statistics for LZ End in Table~\ref{tab:lzend} in the Appendix.}.
Although a phrase of the modified LZ-like encoding is slightly larger than the traditional LZ-like encoding
since it includes the period,
we can see that a much smaller encoding with the same height, which can more than compensate for this increase, can be obtained in some cases. This seemed to be prominent for cere, para, and Escherichia\_Coli, which are all DNA.

\begin{figure}[ht]
    \centerline{
        \includegraphics[width=0.32\textwidth]{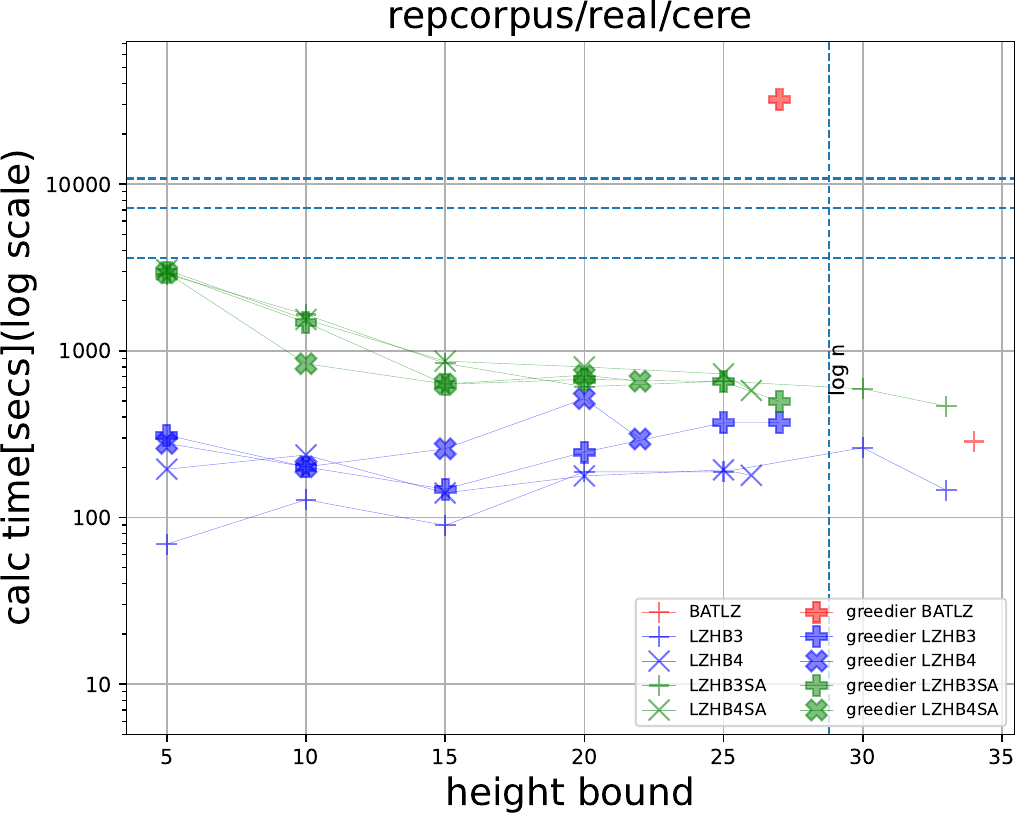}\hfill
        \includegraphics[width=0.32\textwidth]{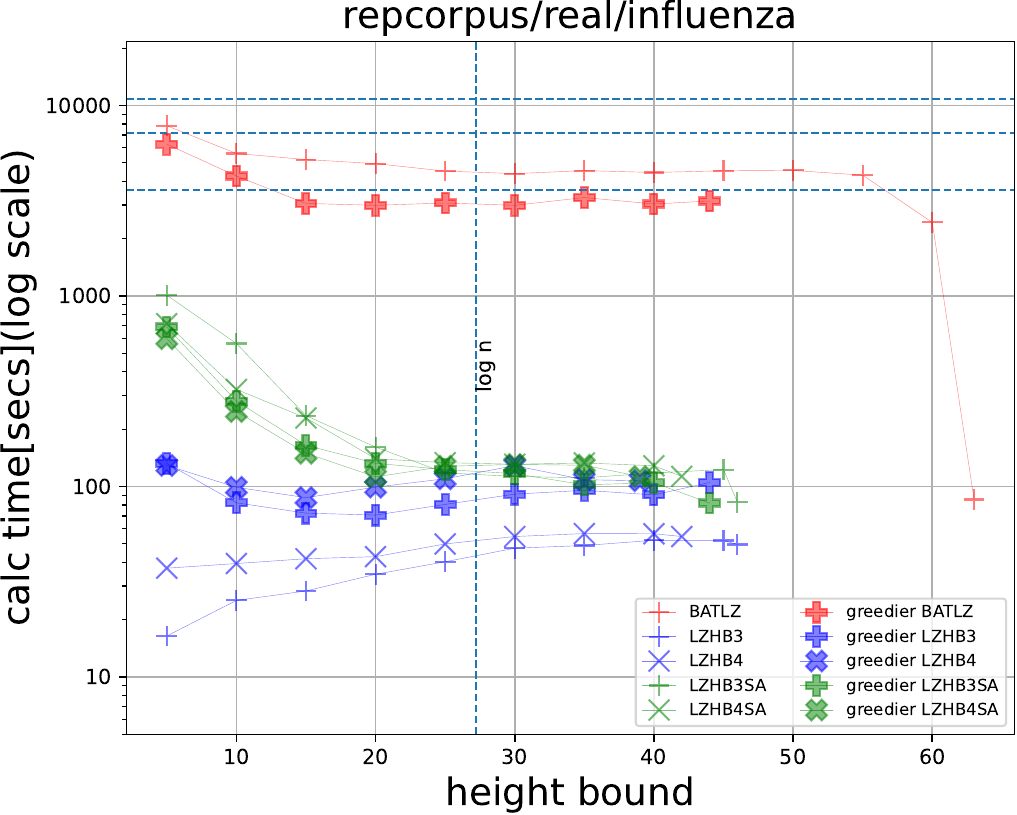}\hfill
        \includegraphics[width=0.32\textwidth]{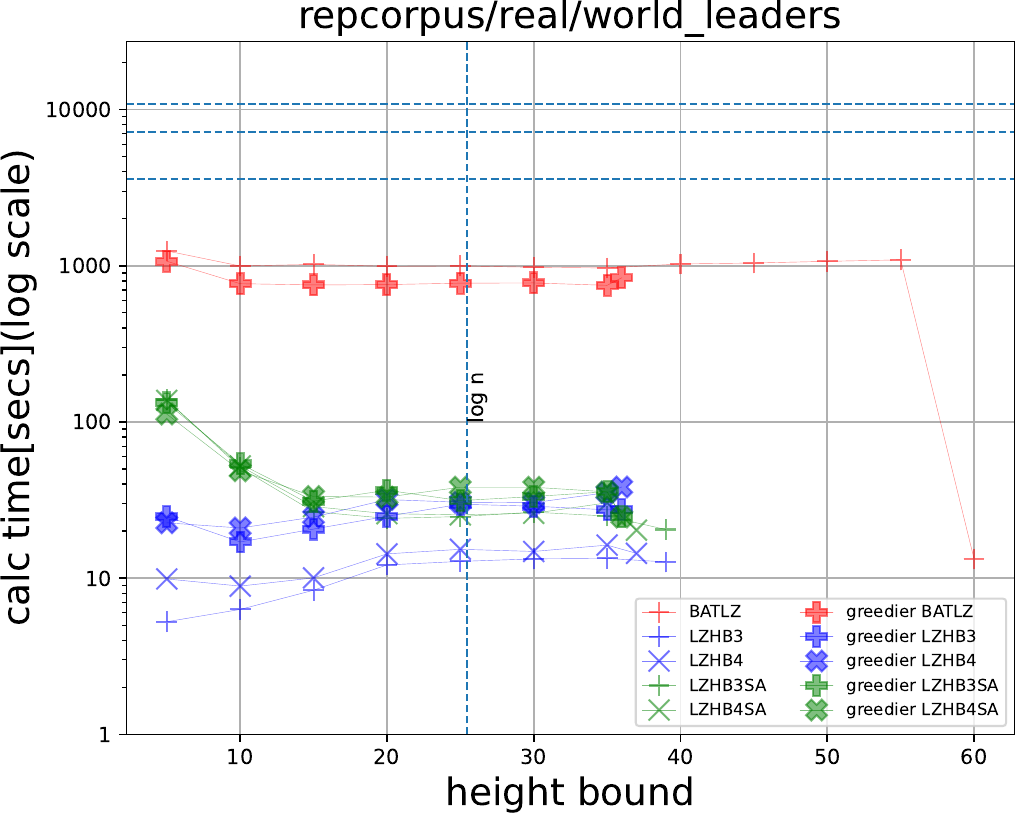}}
    \caption{Running times of \textsf{BAT-LZ} and \textsf{LZHB} variants.
        For each algorithm, the largest height is the height of the parsing obtained with no height constraint.
        Results for smaller height constraints are computed for intervals of $5$, $50$, or $100$, depending on the data.
        Missing plots indicate that the computation exceeded 3 hours.
        See also Fig.~\ref{figure:BATLZ-time-comparison-full} in~\cref{appendix:results}.}
    \label{figure:BATLZ-time-comparison}
\end{figure}
\begin{figure}[ht]
    \centerline{
        \includegraphics[width=0.32\textwidth]{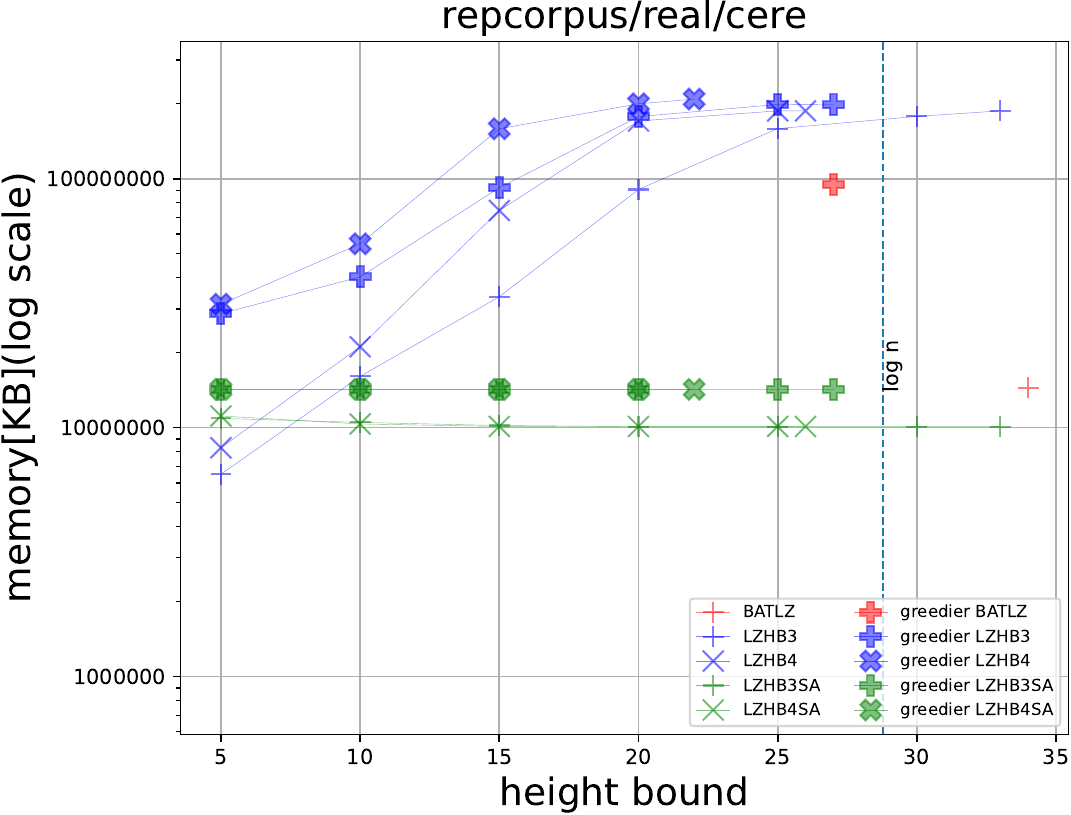}\hfill
        \includegraphics[width=0.32\textwidth]{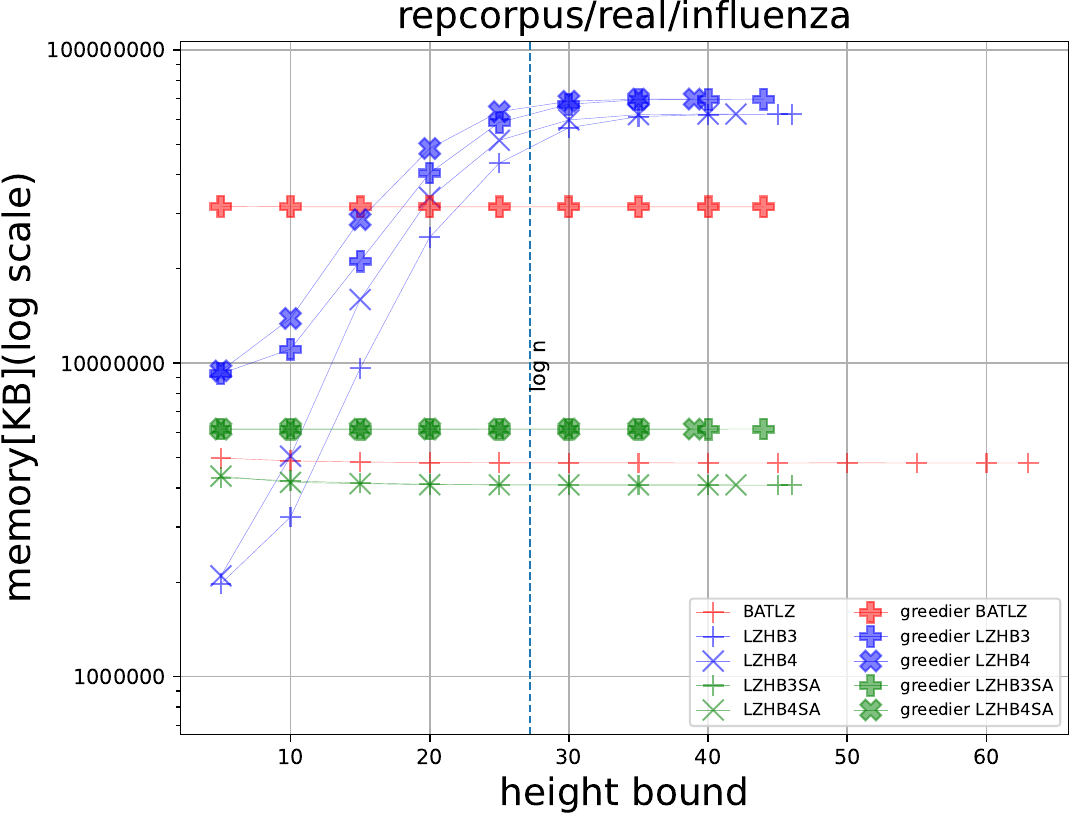}\hfill
        \includegraphics[width=0.32\textwidth]{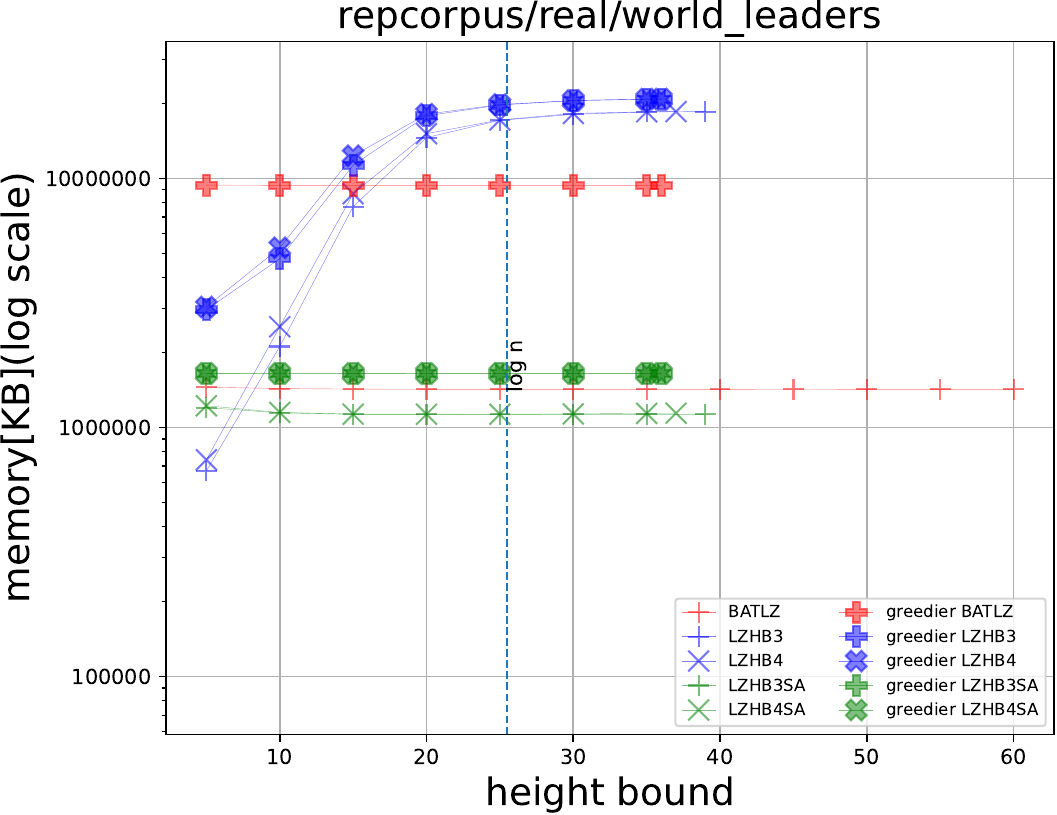}}
    \caption{Memory usage of \textsf{BAT-LZ} and \textsf{LZHB} variants measured by \texttt{getrusage}.
        For each algorithm the largest height is the height of the parsing obtained with no height constraint.
        The results for smaller height constraints are computed for intervals of $5$, $50$, or $100$, depending on the data.
        Missing plots indicate that the computation exceeded 3 hours.
        See Fig.~\ref{figure:BATLZ-space-comparison-full} in~\cref{appendix:results}.}
    \label{figure:BATLZ-space-comparison}
\end{figure}

\begin{figure}[ht]
    \includegraphics[width=0.32\textwidth]{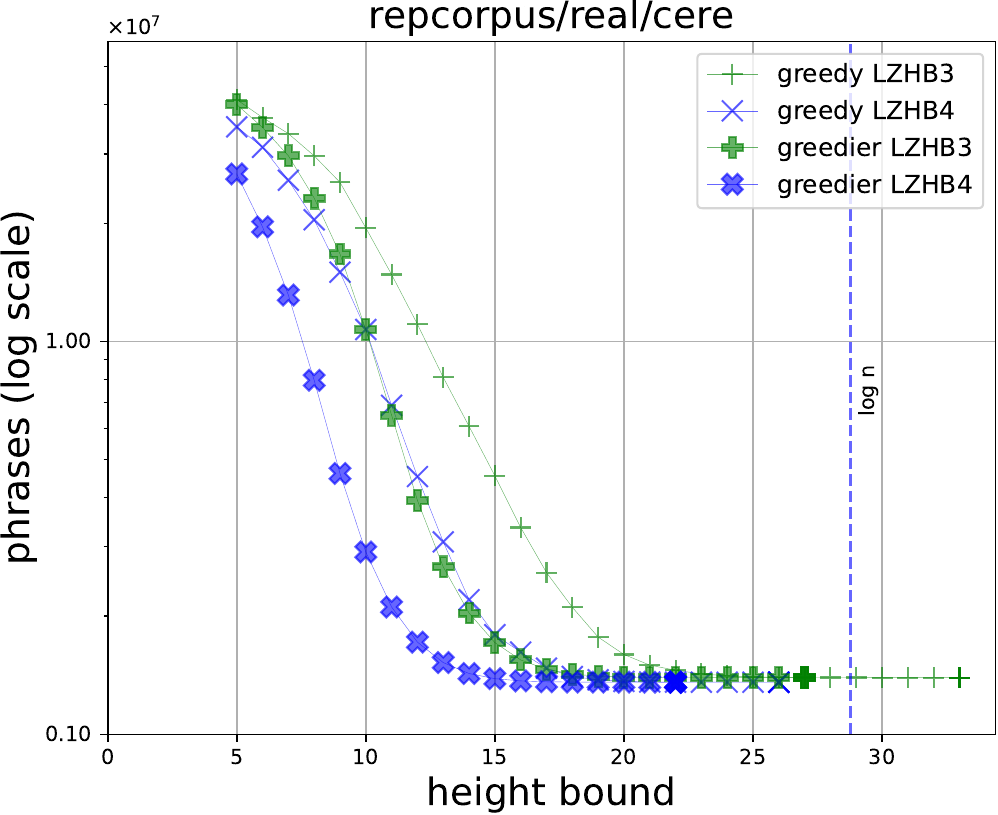}\hfill
    \includegraphics[width=0.32\textwidth]{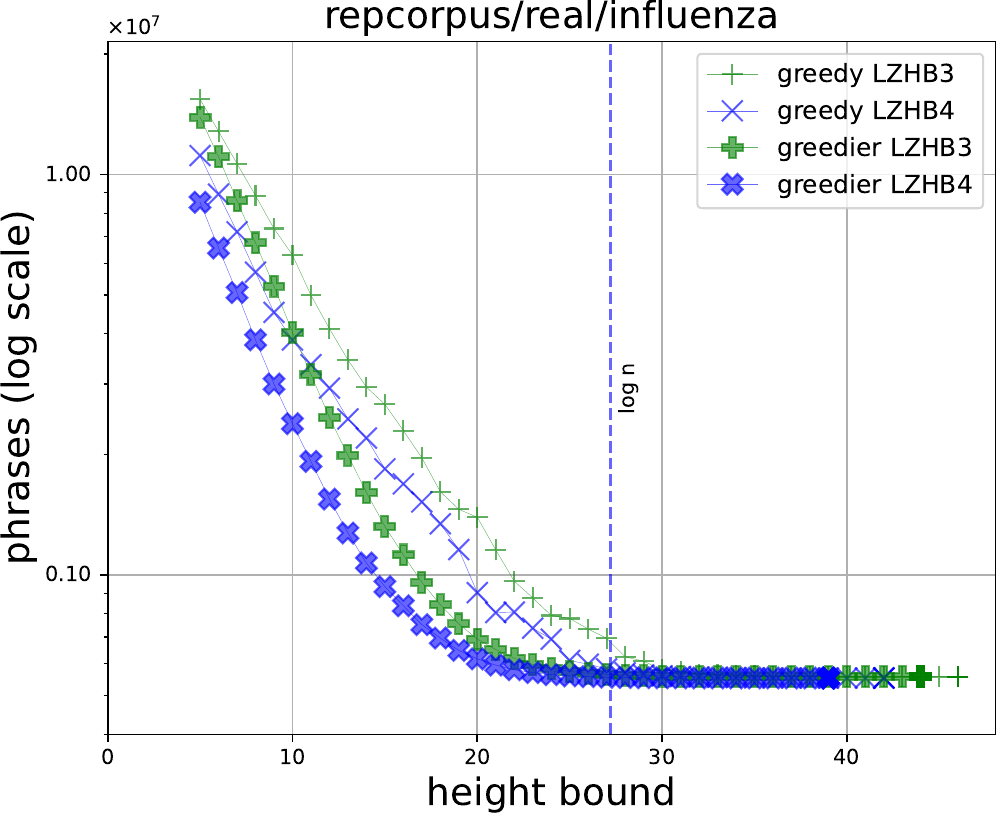}\hfill
    \includegraphics[width=0.32\textwidth]{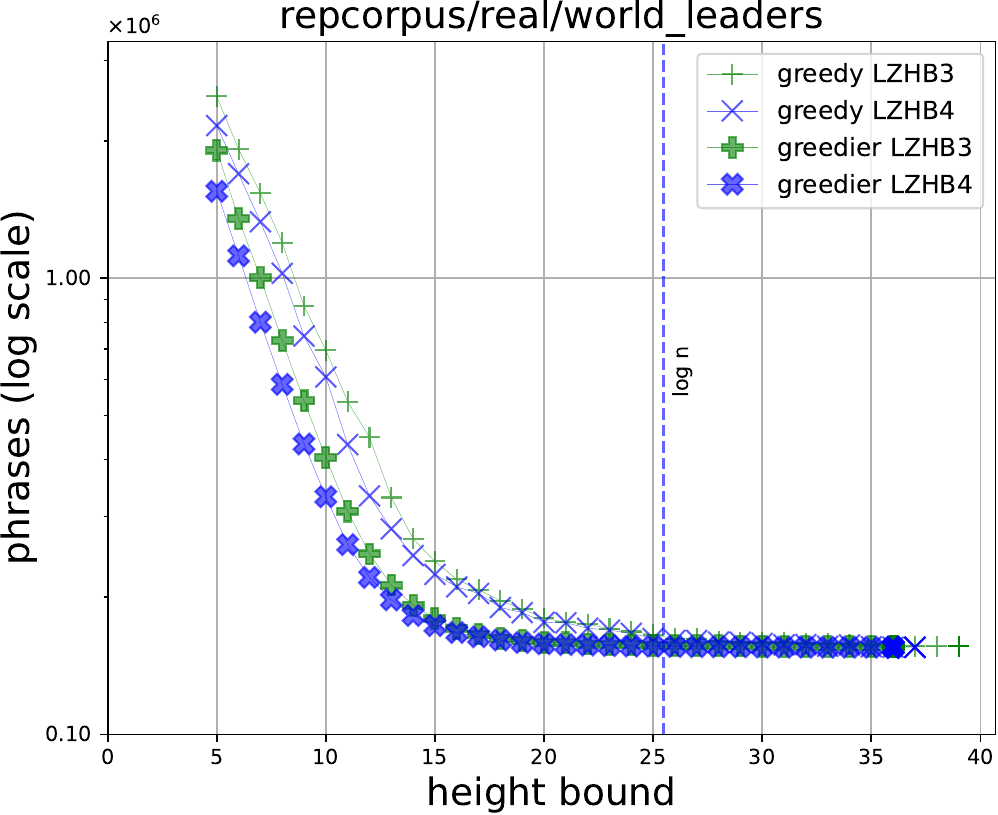}\caption{\# of phrases of \LZHB{3}, \LZHB{4} and their greedier variants.
        See Fig.~\ref{figure:BATLZ-size-comparison-full}~\cref{appendix:results}.}
    \label{figure:BATLZ-size-comparison}
\end{figure}

\section{Discussion}

We have introduced new variants of LZ-like encodings that directly focus on
supporting random access to the underlying data.
We also proposed a modified LZ-like encoding, which naturally connects the run-length encoding
($h=0$), and LZ77 ($h=n$).
We described linear time algorithms for several greedy variants of such encodings, as well as how to adapt
them to support the greedier heuristic.
We showed that our algorithms are faster both theoretically and practically,
compared to contemporaneously proposed algorithms by Lipt\'ak et al.~\cite{liptak2024batlz}.
We also showed some relations between height-bounded encodings
and existing repetitiveness measures.
In particular, we have shown that the optimal LZ-like encoding
with a height constraint of $O(\log n)$,
is one of the asymptotically smallest repetitiveness measures achieving fast (i.e., $O(\polylog (n))$-time) access.

There are numerous avenues future work could take.
An obvious next step is to engineer actual random-access data structures from height-bounded encodings, which our experiments presented here indicate should be competitive with current state-of-the-art methods.
Also, other heuristics for reducing the height while retaining the small size should be explored.

From a theoretical perspective,
a natural open problem is whether we can remove the $\log\sigma$ factor from the running times
of \LZHB{3} and \LZHB{4}.
The biggest open problem we leave is whether better bounds on the size of height-bounded encodings can be established.
Cicalese and Ugazio show that $\lzhbOpt{h}=O(z)$~\cite{cicalese2024complexity} cannot be achieved for costant $h$.
We note that since there exist linear sized data structures that can answer
predecessor queries for a set of values in the range $[1,n]$ in $O(\log \log n)$ time~\cite{DBLP:journals/ipl/Willard83},
even better lower bounds on the height for which $\lzhbOpt{h} = O(z)$ can be shown
by using access time lower bound results of Verbin and Yu~\cite{DBLP:conf/cpm/VerbinY13}.
However, these do not rule out $O(\polylog(n))$ height.
Can we {\em balance} the height (achieve $O(\polylog(n))$ height)
of an LZ-like encoding or a modified LZ-like encoding, by only increasing the size by a constant factor?

\clearpage
\bibliographystyle{plainurl}
\bibliography{refs}

\clearpage
\appendix

\section{Ommitted Proofs}\label{appendix:claims}
\begin{observation}\label{observation:linear-height}
    The height of the optimal LZ-like encoding (LZ77) can be $\Theta(n)$.
\end{observation}
\begin{proof}
    Consider the height of $\mathtt{b}$ in the string: \[\mathtt{ababcdbc\#dbefbe\#fbghbg\#\dots}\#
        c_k\mathtt{b}c_{k+1}c_{k+2}\mathtt{b}c_{k+1}\#
        c_{k+2}\mathtt{b}c_{k+3}c_{k+4}\mathtt{b}c_{k+3}\#\dots\]
    It is not difficult to see each $\mathtt{b}$ will reference the closest previously occurring $\mathtt{b}$.
\end{proof}

\lowerboundGrl*
\begin{proof}
    It is known that a given RLSLP can be {\em balanced}, i.e., transformed into another RLSLP whose parse tree height is $O(\log n)$, without changing the asymptotic size of the RLSLP~\cite{DBLP:conf/spire/NavarroOU22}.

    \begin{theorem}[Theorem~1 in~\cite{DBLP:conf/spire/NavarroOU22}]\label{theorem:balancing_rlg}
        Given an RLSLP $G$ generating a string $w$, it is possible to construct an equivalent balanced RLSLP $G'$ of size $O(|G|)$, in linear time, with only rules of the form $A\rightarrow a$,$A\rightarrow BC$, and $A\rightarrow B^t$, where $a$ is a terminal, $B$ and $C$ are variables, and $t > 2$.
    \end{theorem}

    Given an RLSLP with parse tree height $h$, a partial parse tree of the RLSLP, as defined in Theorem 5 of~\cite{DBLP:journals/tit/NavarroOP21}, can be obtained.
    Based on this partial parse tree, it is possible to obtain an LZ-like encoding of size at most that of the run-length grammar.
    The height of the LZ-like encoding corresponds to the height of the parse tree, and thus is $h$-bounded.
    Since, from Theorem~\ref{theorem:balancing_rlg}, we can obtain a grammar of size $O(\rlgOpt)$ with height $O(\log n)$,
    we have that the size $\lzhbOpt{c\log n}$ of an optimal $c\log n$-bounded LZ-like encoding
    will have size at most $O(\rlgOpt)$,
    where $c$ is a constant incurred by the height of the balanced RLSLP obtained from balancing the optimal run-length grammar.
\end{proof}

\lowerboundGrlAndGit*
\begin{proof}
    The proof follows similar arguments to those of Bille et al.~\cite{DBLP:journals/jda/BilleGGP18}.
    In their proof, they consider a string (originally used by Charikar et al.~\cite{DBLP:journals/tit/CharikarLLPPSS05})
    \[\hat{s} = t(k_1) \#_1 t(k_2) \#_2 \cdots \#_{p-1} t(k_p)\]
    where $t(k)$ denotes the length-$k$ prefix of the Thue-Morse word,
    $\#_k$ are distinct symbols, and
    $k_1, k_2, \ldots, k_p$ is an integer sequence where
    $k_1$ is the largest of the $k_i$, $p = \Theta(\log k_1)$,
    and show that the integers $k_i$ can be chosen so that the smallest RLSLP for $\hat{s}$ has size $\Omega(\frac{\log^2{k_1}}{\log\log k_1})$.
    They further observe that the size of an LZ77 parsing of the length-$k$ prefix of the Thue-Morse word is $O(\log k)$~\cite{DBLP:conf/mfcs/BerstelS06,DBLP:journals/siamdm/ConstantinescuI07}. Since $k_1$ is the largest of the $k_i$'s,
    the size of the LZ77 parsing of $\hat{s}$ is $O(\log k_1 + p) = O(\log k_1)$ for this class of strings.

    Now, observe that the height, as defined in Equation~(\ref{eq:height:period}), of an LZ-like encoding with size $z'$ is $O(z')$, since any position in a phrase references a position in a previous phrase. Therefore, we have that
    for some $c$, $\lzhbOpt{c\log n} = O(\log k_1)$ as well, thus showing that $\lzhbOpt{c\log n} = o(\rlgOpt)$.

    Since Thue-Morse words are cube-free,
    the extra rules in the iterated SLPs do not improve the asymptotic size compared to RLSLPs,
    and thus, as shown in Lemma 3 of~\cite{NavarroUrbinaLATIN2024},
    $g_{it(d)}=\rlgOpt$ for this family of strings. Therefore, we have
    $\lzhbOpt{c\log n} = o(g_{it(d)})$.
\end{proof}

\twoApproximation*

\begin{proof}
    It is obvious that $\hat{\tilde{z}}\leq \tilde{z}$.
    We first claim $\tilde{z}\leq z$.
    This can be seen from the fact that a phrase that starts at position $i$ in (the greedy) \LZHB{4} is at least as long as $\lpf(i)$,
    and so any LZHB phrase must extend at least to the end of the LZ77 phrase it starts in.
    Thus, for each LZHB phrase, we can assign a distinct LZ77 phrase (the LZ77 phrase it starts in),
    and therefore the number of \LZHB{4} phrases cannot be more than the number of LZ77 phrases.

    Next, we claim $z \leq 2\tilde{z}'$ for any modified LZ-like encoding
    with no height constraints, whose size is denoted by $\tilde{z}'$.
    This can be seen from the fact that a given phrase in a modified LZ-like encoding can contain at most two LZ77 phrases that start in it.
    Let $x^kx'$ be a phrase in the modified LZ-like encoding starting at position $i$, whose minimum period is $|x|$ and $x'$ is a prefix of $x$.
    Since there is a previous occurrence of $x$, the first LZ77 phrase that starts in the  LZHB phrase must extend to at least the end of the first $x$.
    Since there is a previous occurrence of any suffix of $x^{k-1}x'$, the next LZ77 phrase will extend at least to the end of the LZHB phrase.
    The theorem follows from the above arguments, using $\hat{\tilde{z}}$ for $\tilde{z}'$.
\end{proof}

\clearpage
\section{Pseudo codes}\label{appendix:pseudocode}
\begin{algorithm2e}[htbp]
    \caption{$O(n\log\sigma)$ time algorithm for computing \LZHB{3}}
    \label{alg:lz77h3_selfreferencing2}
    \SetVlineSkip{0.5mm}
    \Fn{$\LZHB{3}(T,h)$}{
    Let $\mathcal{T}$ be a suffix tree of $\varepsilon$\;
    Initialize $H, \mathcal{E}$\;
    $i = 1, n = |T|, \mathit{t} = 0$\;
    \While{$i \le n$}{
    \tcp{$T[i..i+\ell)$ is longest prefix of $T[i..n]$ occurring in $T'[1..i)$,}
    \tcp{and leftmost occurrence is $s$.}
    $(\ell, s) = \mathcal{T}.\mathit{prefq}(T[i..n])$\;

    \If(\tcp*[f]{check for longer self-referencing occurrence}){$\ell > 0$}{
    $k$ = leftmost occurrence of $T[i..i+\ell)$
        in $T[\max(\mathit{t},i-\ell)..i+\ell-1)$\;
        \If{k != \textsf{nil}}{
            \lIf{$T[k+\ell] = T[i+\ell]$}{$s = k$}
            \lWhile{$T[s+\ell] = T[i+\ell]$}{$\ell = \ell + 1$}
        }
        }

        \eIf{$\ell \leq 1$}{
        Add $(1,T[i])$ to $\mathcal{E}$\;
    $H[i] = 0$\;
        \lIf{h > 0}{
            $\mathcal{T}.\mathit{append}(T[i])$
        }
        \lElse{
            $\mathcal{T}.\mathit{append}(\$)$
        }
    $i = i + 1$\;
        }{
    $S = T[i..i+\ell)$\;
        Add $(\ell,s)$ to $\mathcal{E}$\;
        \For{$j = 0$ \textbf{\textup{to}} $\ell-1$}{
            $H[i+j] = H[s+(j\bmod (i - s))]+1$\;
            \If{$H[i+j] = h$}{
                $S[j+1] = \$$\;
                $t = i + j$\;
            }
        }
    $\mathcal{T}.\mathit{append}(S)$\;
    $i = i + \ell$\;
        }
        }
        \textbf{return} $\mathcal{E}$\;
    }
\end{algorithm2e}

\clearpage
\section{More results of computational experiments}\label{appendix:results}

\begin{figure}[htbp]
    \includegraphics[width=0.49\textwidth]{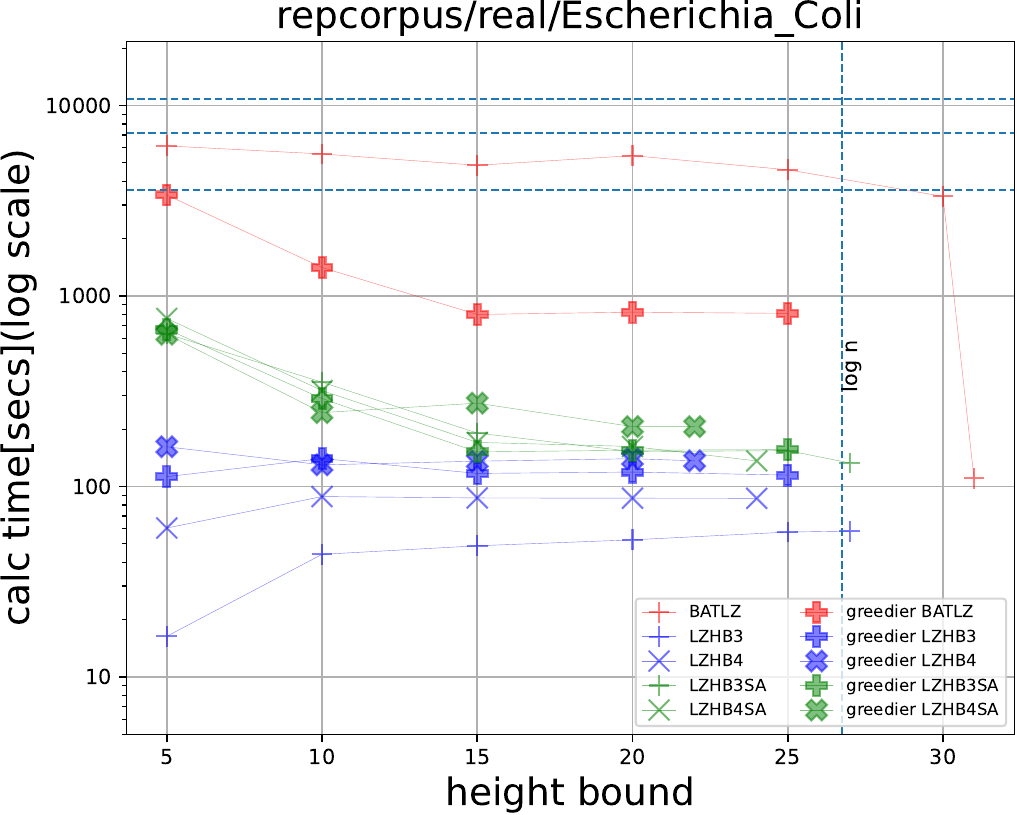}\hfill
    \includegraphics[width=0.49\textwidth]{calc_time_cere.pdf}\\

    \medskip

    \includegraphics[width=0.49\textwidth]{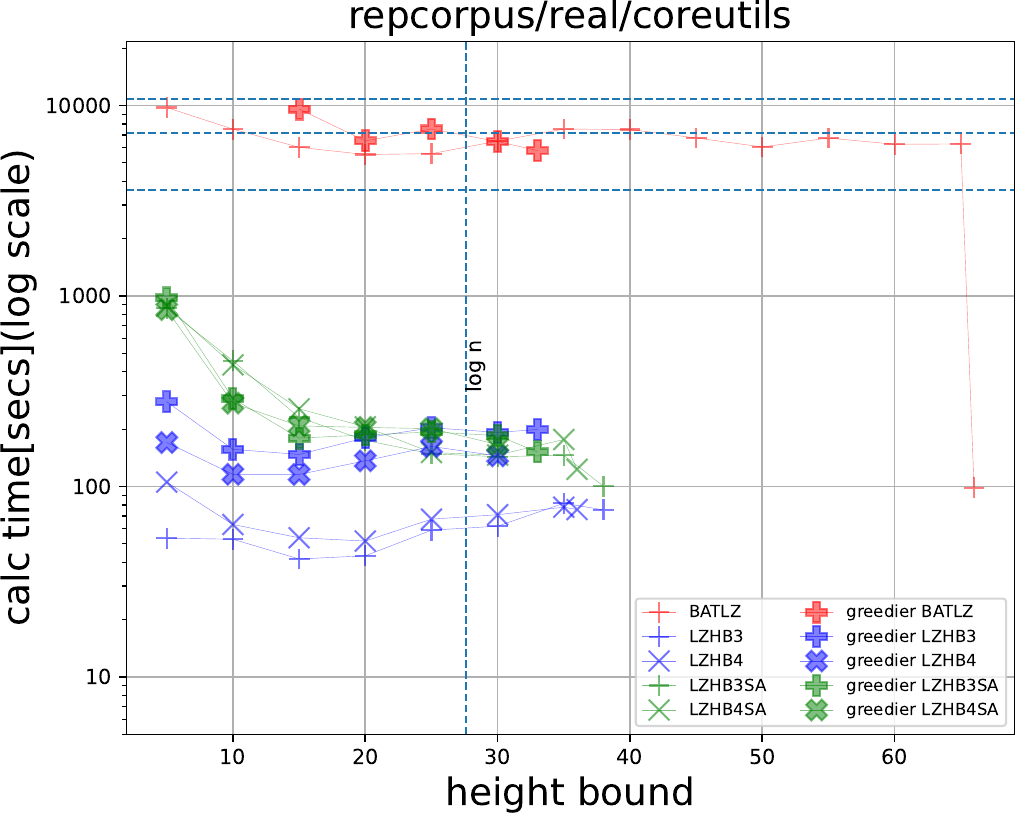}\hfill
    \includegraphics[width=0.49\textwidth]{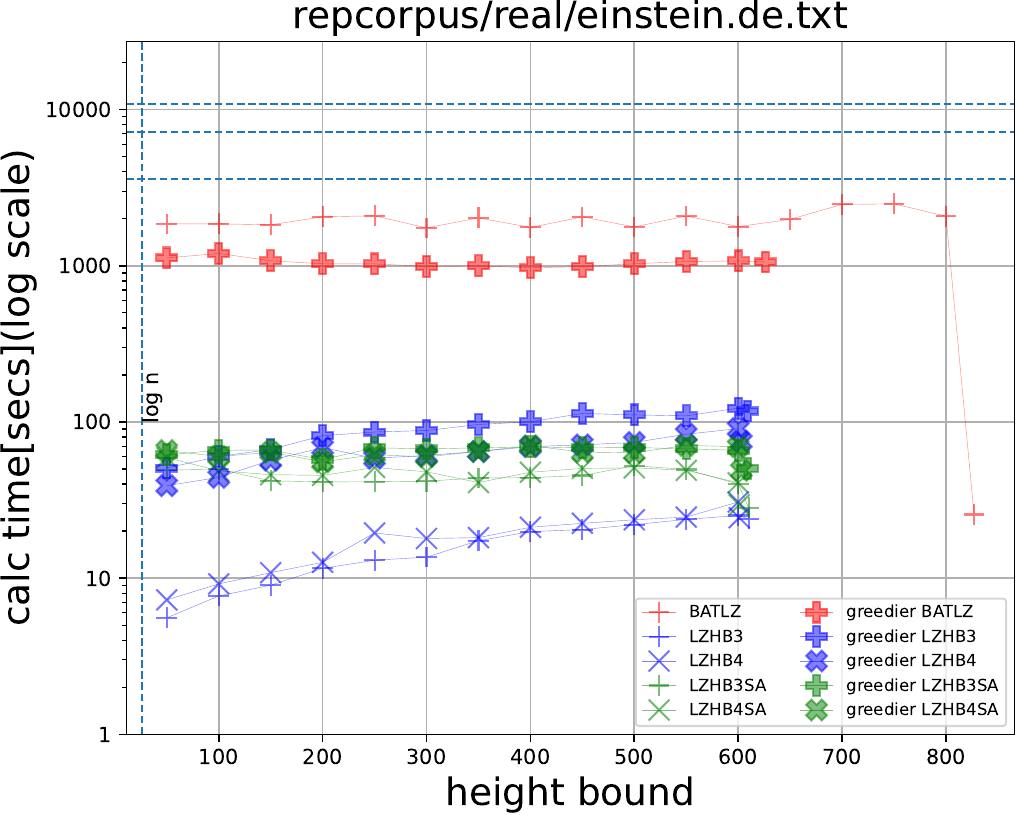}
    \caption{Running times of \textsf{BAT-LZ} and \textsf{LZHB} variants.
        The dotted horizontal lines respectively correspond to 1,2,3 hours.
        For each algorithm the largest height is the height of the parsing obtained with no height constraint, which is always computed.
        The results for smaller height constraints are computed for intervals of $5$, $50$, or $100$, depending on the data.
        Missing plots indicate that the computation exceeded 3 hours, which only happened for the \textsf{BAT-LZ} variants,
        especially prominent in cere, para, and einstein.en.txt.
    }
    \label{figure:BATLZ-time-comparison-full}
\end{figure}

\begin{figure}[htbp]
    \ContinuedFloat
    \includegraphics[width=0.49\textwidth]{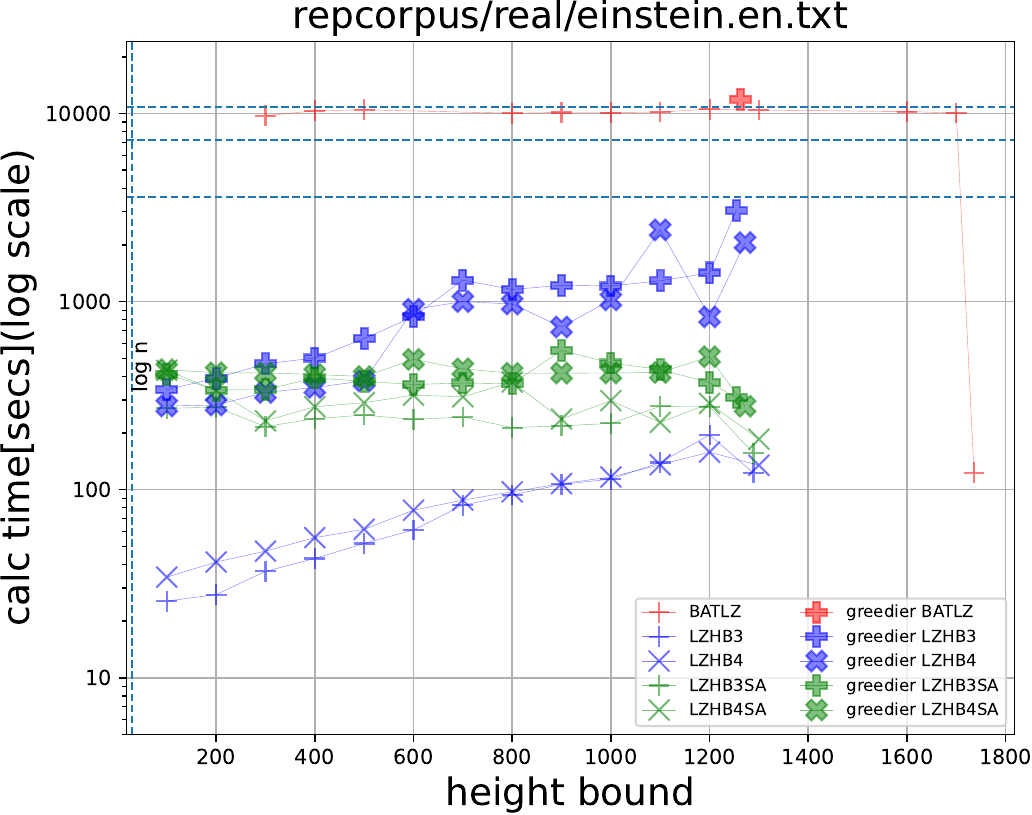}\hfill
    \includegraphics[width=0.49\textwidth]{calc_time_influenza.pdf}\\

    \medskip

    \includegraphics[width=0.49\textwidth]{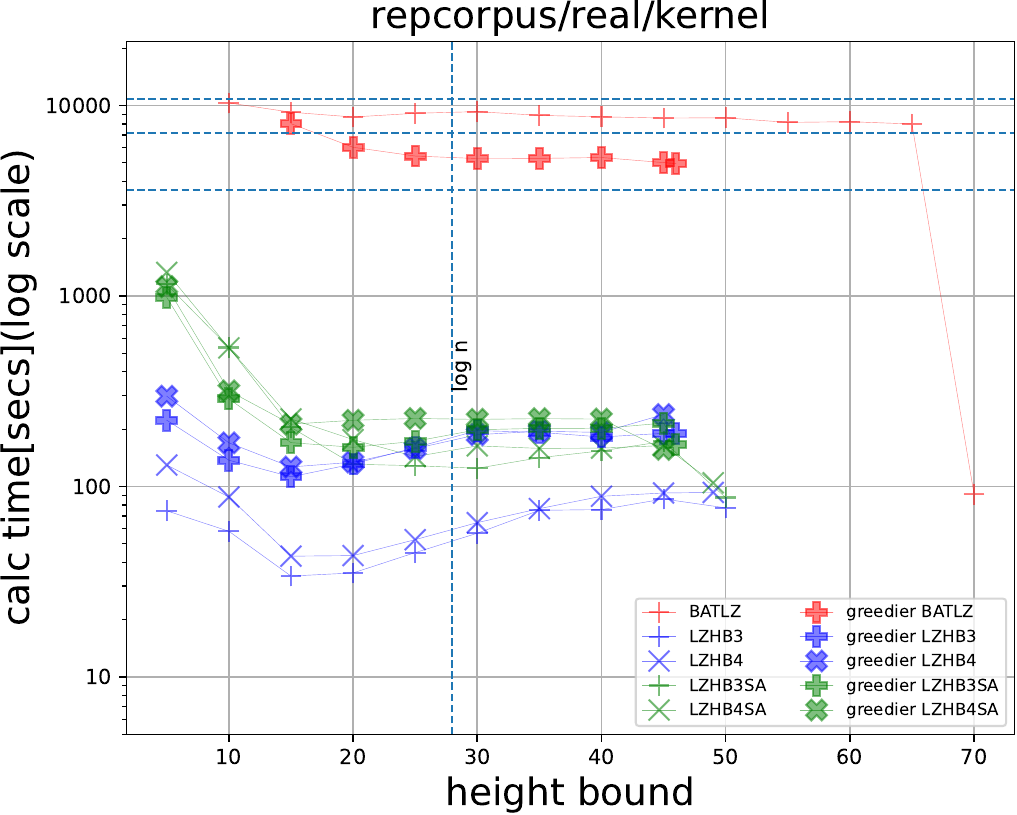}\hfill
    \includegraphics[width=0.49\textwidth]{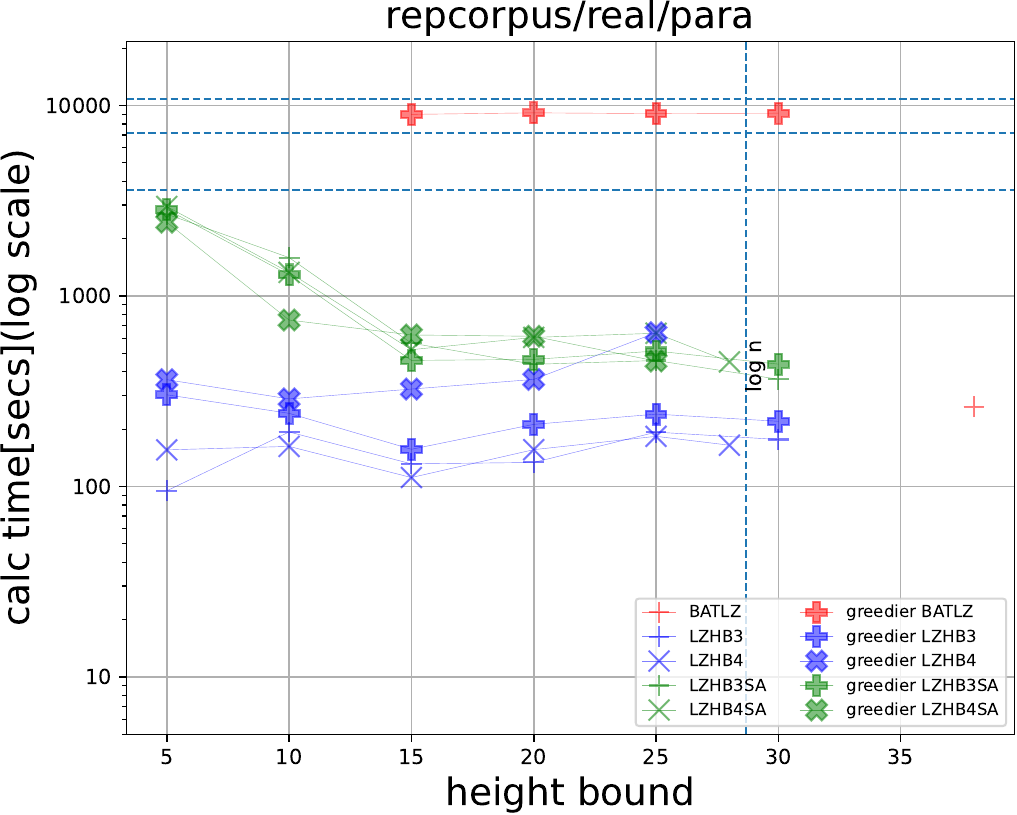}\\

    \medskip

    \includegraphics[width=0.49\textwidth]{calc_time_world_leaders.pdf}
    \caption{(contd.)
        Running times of \textsf{BAT-LZ} and \textsf{LZHB} variants.
        The dotted horizontal lines respectively correspond to 1,2,3 hours.
        For each algorithm the largest height is the height of the parsing obtained with no height constraint, which is always computed.
        The results for smaller height constraints are computed for intervals of $5$, $50$, or $100$, depending on the data.
        Missing plots indicate that the computation exceeded 3 hours, which only happened for the \textsf{BAT-LZ} variants,
        especially prominent in cere, para, and einstein.en.txt.}
\end{figure}

\begin{figure}
    \includegraphics[width=0.49\textwidth]{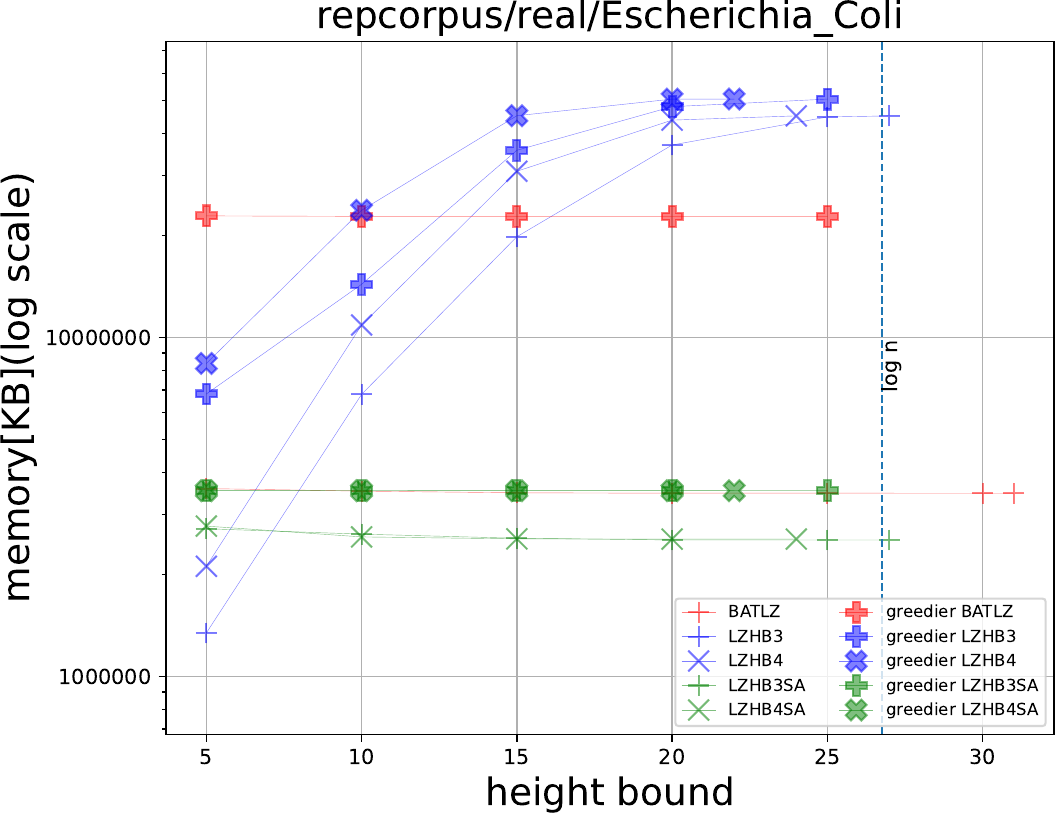}\hfill
    \includegraphics[width=0.49\textwidth]{calc_space_cere.pdf}\\

    \medskip

    \includegraphics[width=0.49\textwidth]{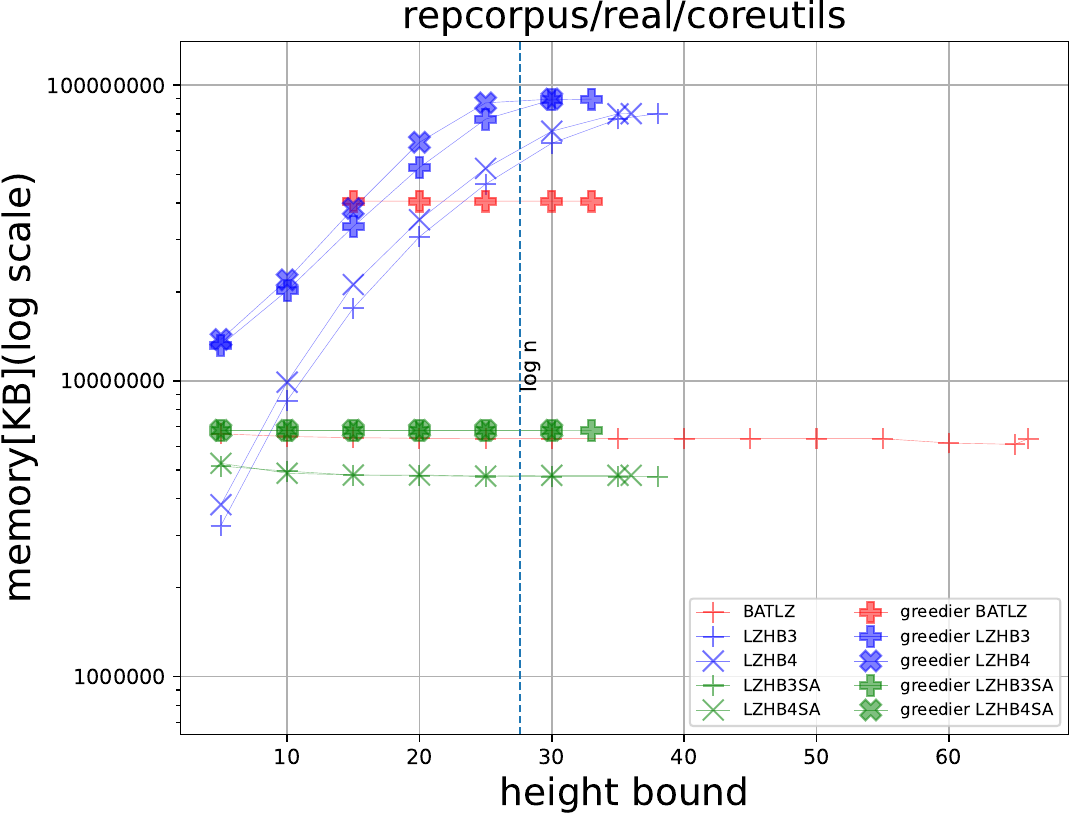}\hfill
    \includegraphics[width=0.49\textwidth]{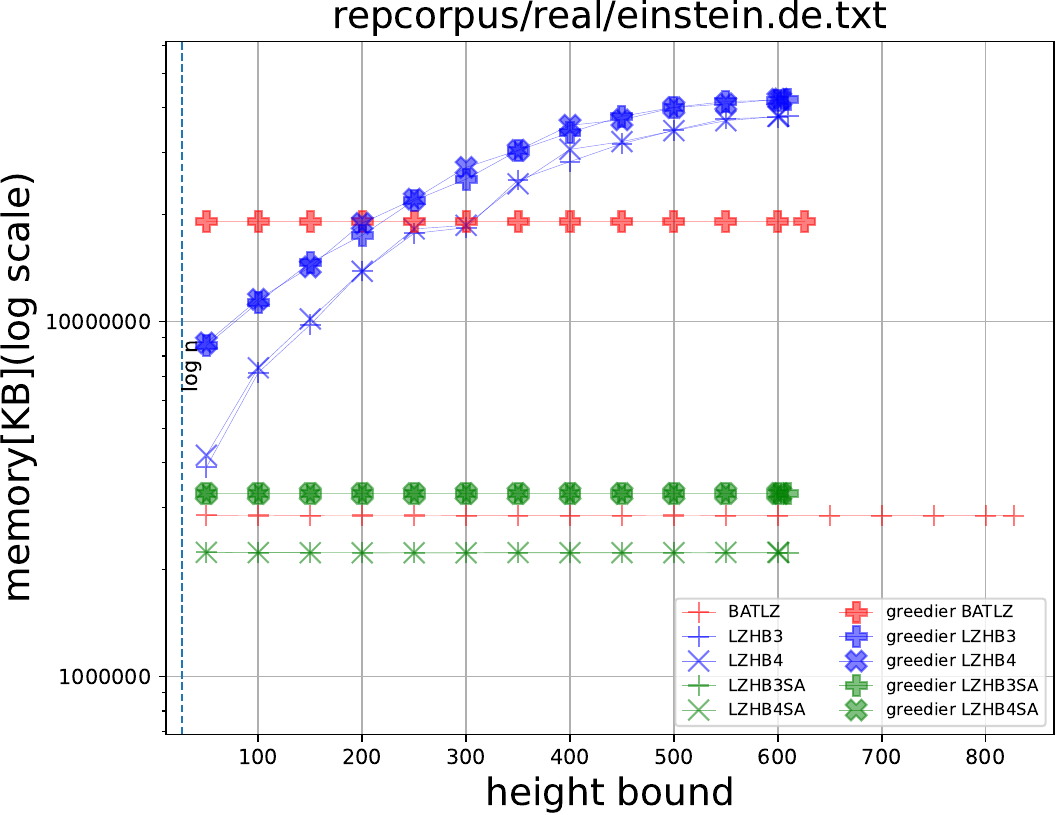}
    \caption{Memory usage of \textsf{BAT-LZ} and \textsf{LZHB} variants.
        For each algorithm the largest height is the height of the parsing obtained with no height constraint, which is always computed.
        nn    The results for smaller height constraints are computed for intervals of $5$, $50$, or $100$, depending on the data.
        Missing plots indicate that the computation exceeded 3 hours, which only happened for the \textsf{BAT-LZ} variants,
        especially prominent in cere, para, and einstein.en.txt.}
    \label{figure:BATLZ-space-comparison-full}
\end{figure}

\begin{figure}
    \ContinuedFloat
    \includegraphics[width=0.49\textwidth]{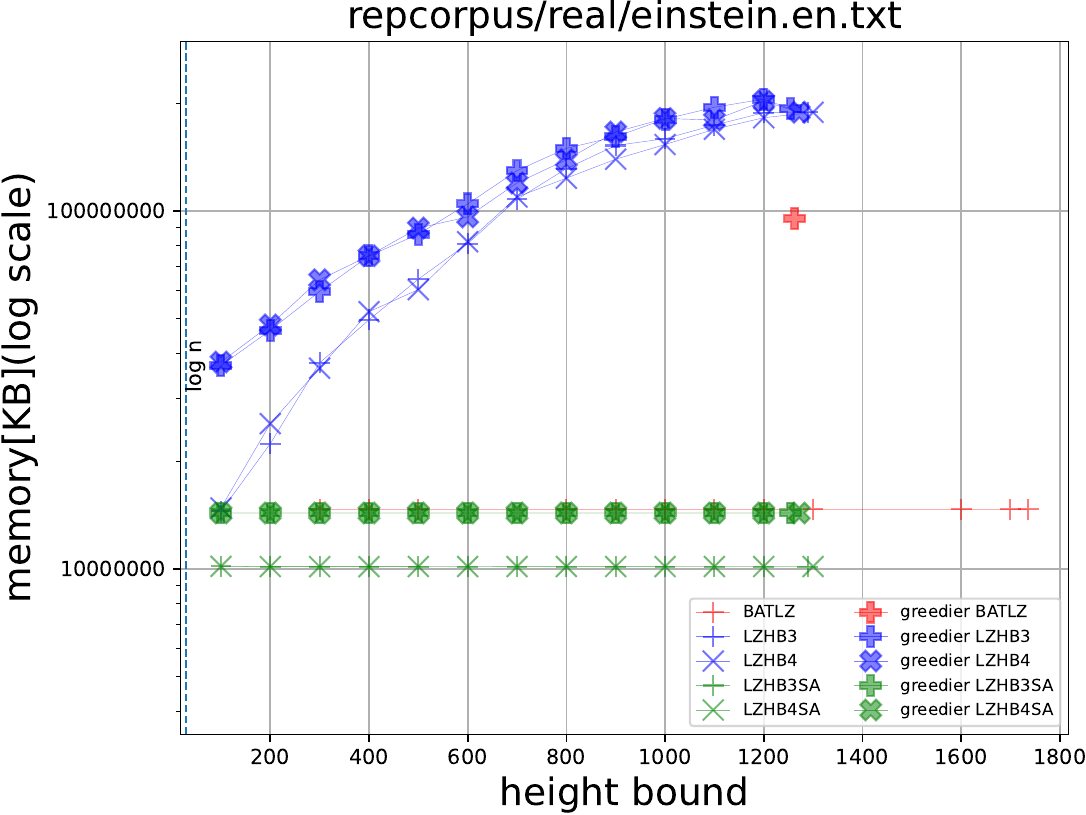}\hfill
    \includegraphics[width=0.49\textwidth]{calc_space_influenza.pdf}\\

    \medskip

    \includegraphics[width=0.49\textwidth]{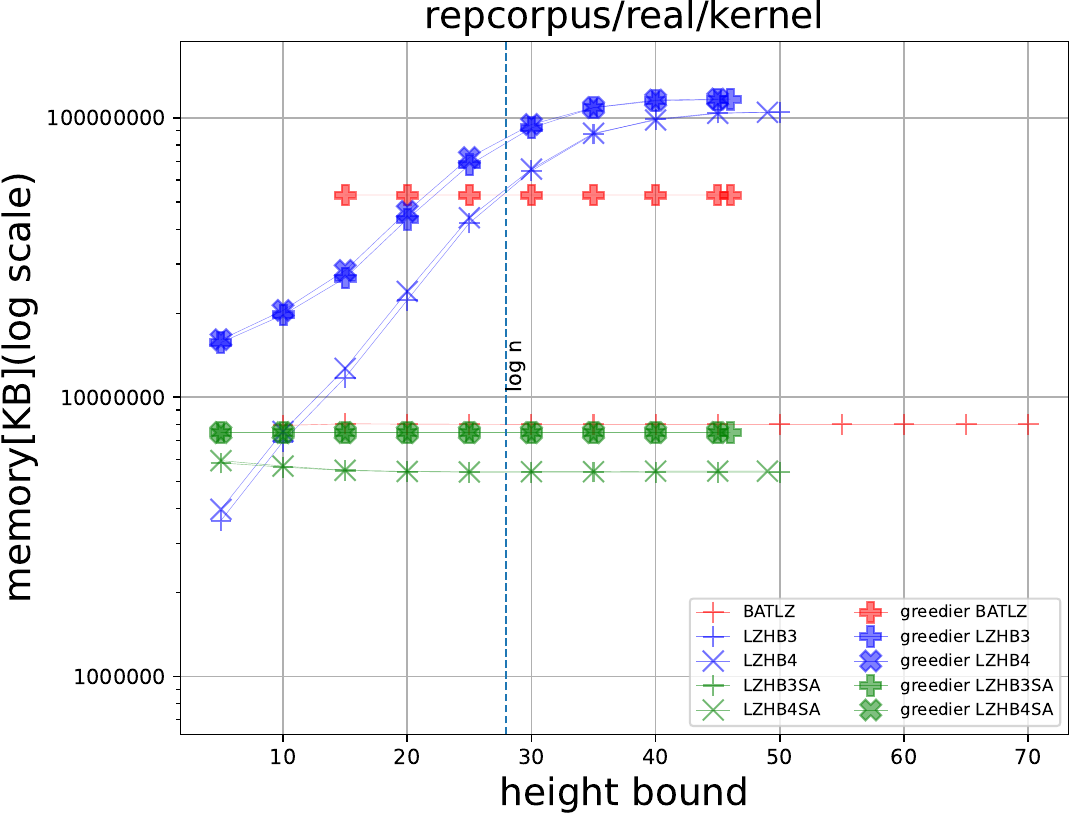}\hfill
    \includegraphics[width=0.49\textwidth]{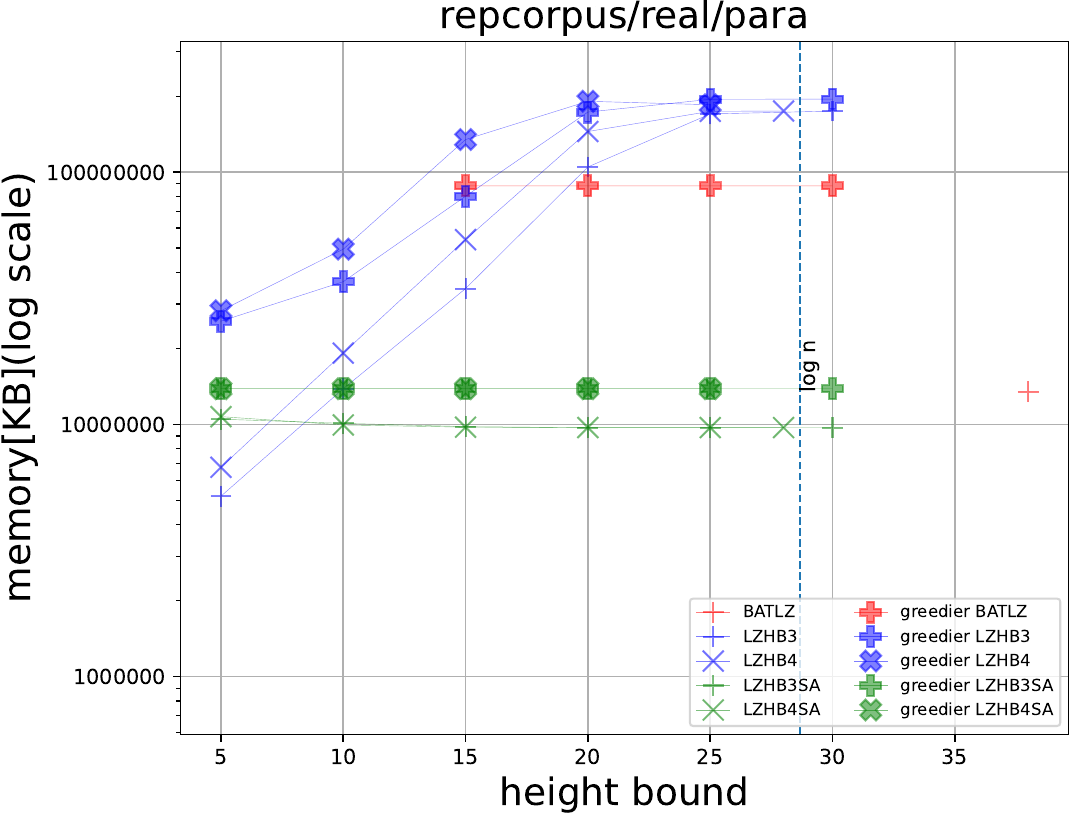}\\

    \medskip

    \includegraphics[width=0.49\textwidth]{calc_space_world_leaders.pdf}
    \caption{(contd.) Memory usage of \textsf{BAT-LZ} and \textsf{LZHB} variants.
        For each algorithm the largest height is the height of the parsing obtained with no height constraint, which is always computed.
        The results for smaller height constraints are computed for intervals of $5$, $50$, or $100$, depending on the data.
        Missing plots indicate that the computation exceeded 3 hours, which only happened for the \textsf{BAT-LZ} variants,
        especially prominent in cere, para, and einstein.en.txt.}
\end{figure}

\begin{figure}
    \includegraphics[width=0.49\textwidth]{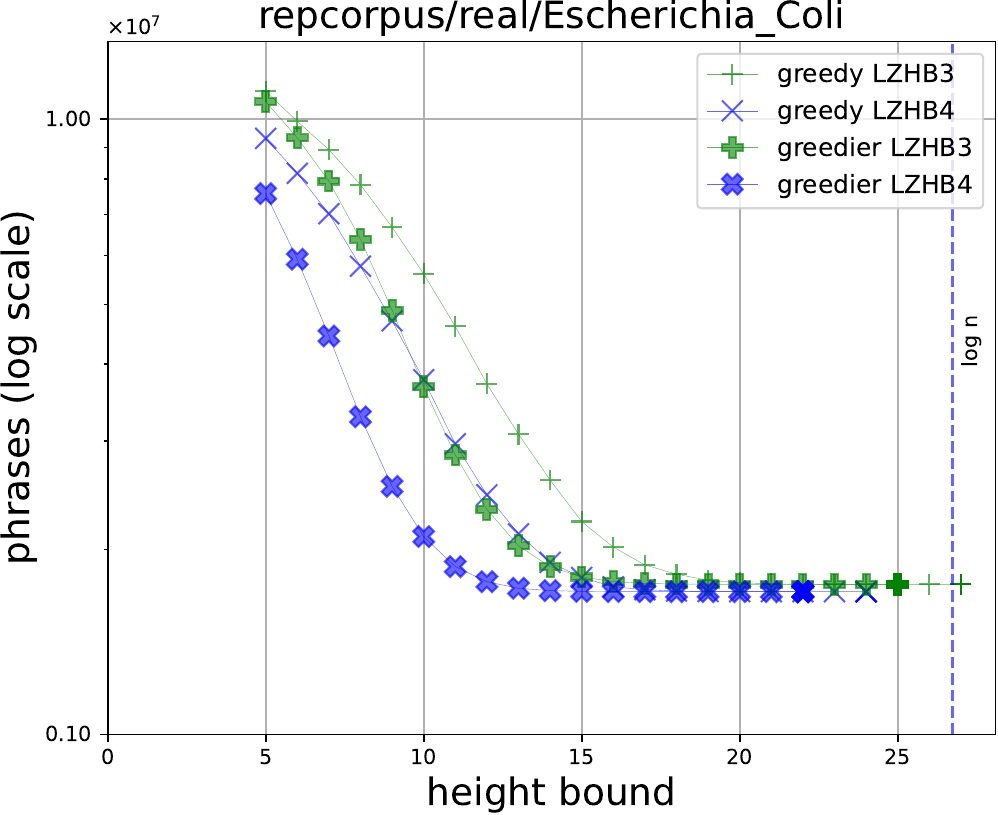}\hfill
    \includegraphics[width=0.49\textwidth]{calc_cmp_size_sa_cere.pdf}\\

    \medskip

    \includegraphics[width=0.49\textwidth]{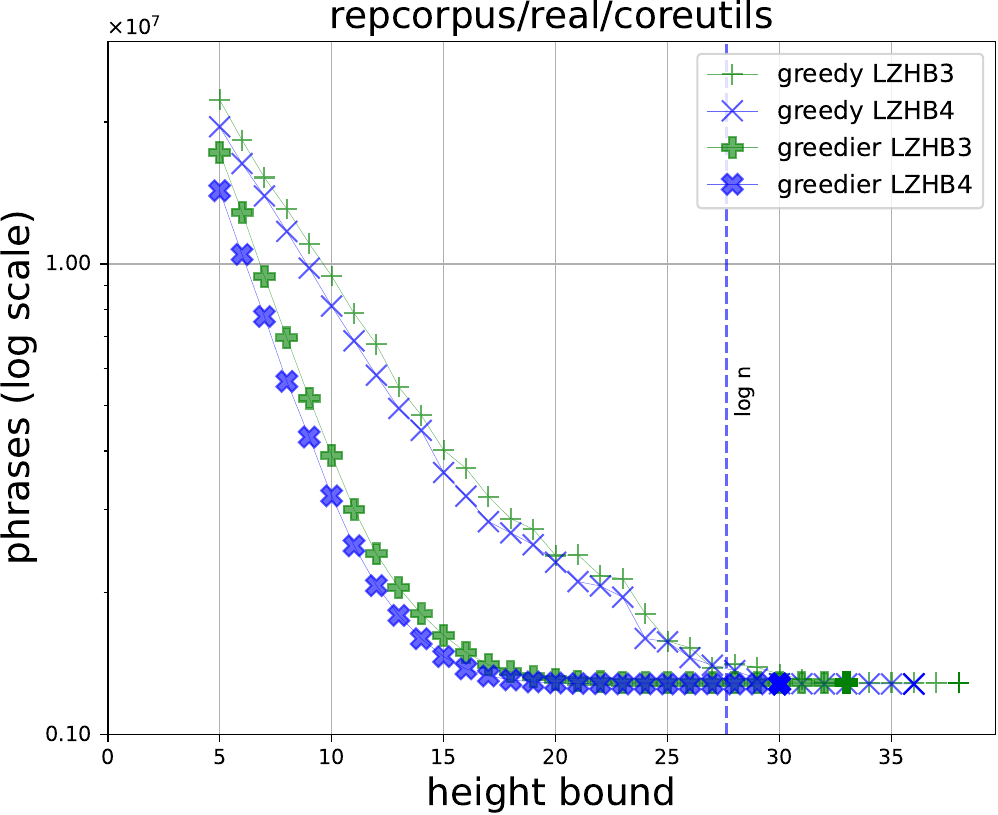}\hfill
    \includegraphics[width=0.49\textwidth]{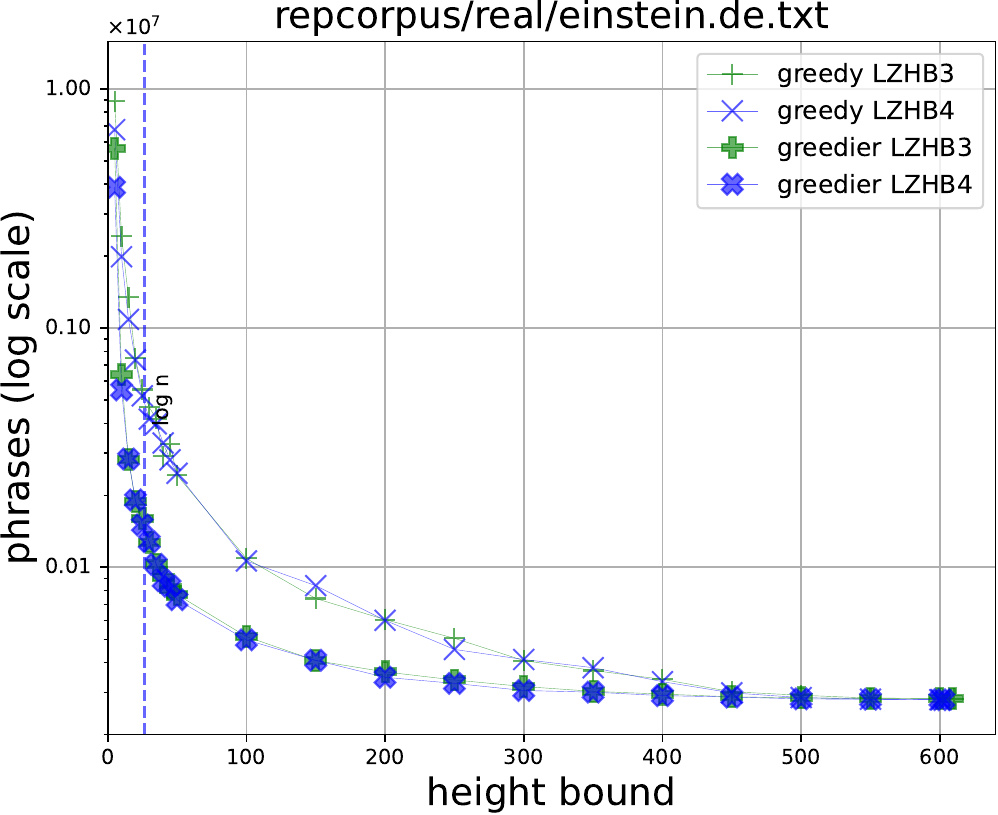}
    \caption{Number of phrases of \LZHB{3} and \LZHB{4} and their greedier variants.}
    \label{figure:BATLZ-size-comparison-full}
\end{figure}
\begin{figure}
    \ContinuedFloat
    \includegraphics[width=0.49\textwidth]{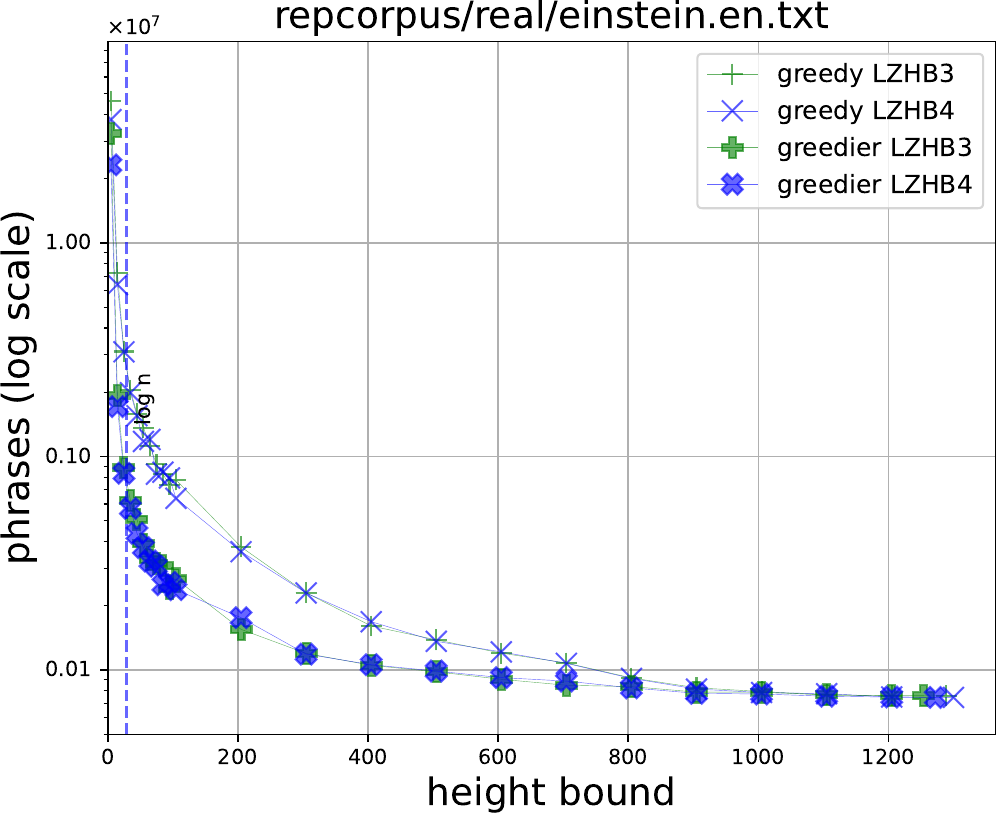}\hfill
    \includegraphics[width=0.49\textwidth]{calc_cmp_size_sa_influenza.pdf}\\

    \medskip

    \includegraphics[width=0.49\textwidth]{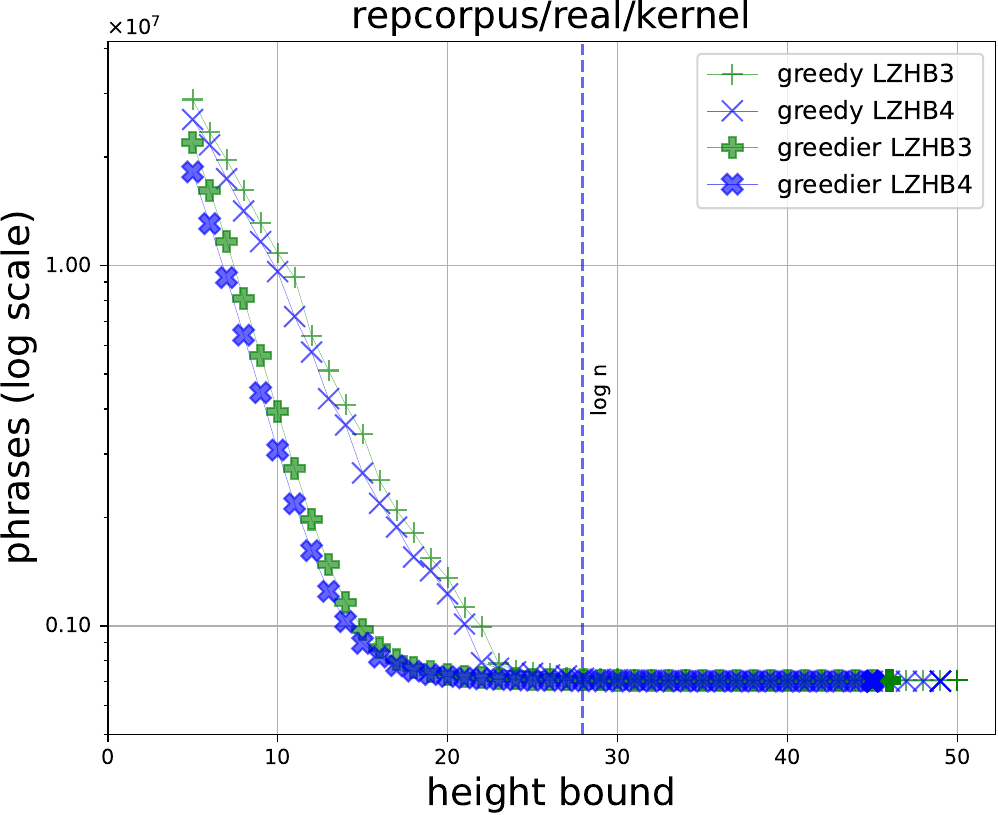}\hfill
    \includegraphics[width=0.49\textwidth]{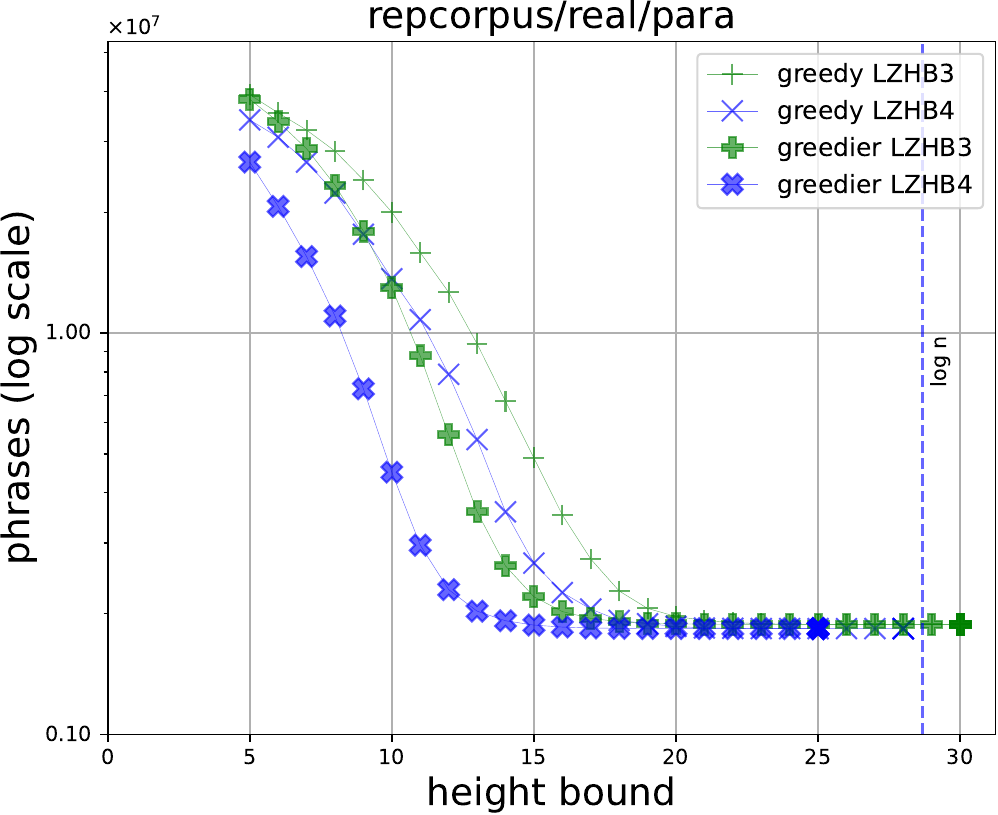}\\

    \medskip

    \includegraphics[width=0.49\textwidth]{calc_cmp_size_sa_world_leaders.pdf}
    \caption{(contd.) Number of phrases of \LZHB{3} and \LZHB{4} and their greedier variants.}
\end{figure}

\begin{table}[ht]

    \caption{The number of factors (phrases) and the maximum height of the LZ End parsing for each of our data sets.}
    \label{tab:lzend}
    \centering
    \begin{tabular}{||l r r||}
        \hline
        Dataset           & Max. Height & \# Factors \\ [0.5ex]
        \hline\hline
        Escherichia\_Coli & 27          & 2212539    \\
        cere              & 256         & 1863246    \\
        coreutils         & 174         & 1555394    \\
        einstein.de.txt   & 60          & 39587      \\
        einstein.en.txt   & 117         & 104087     \\
        influenza         & 79          & 919565     \\
        kernel            & 45          & 868362     \\
        para              & 258         & 2539381    \\
        world\_leaders    & 102         & 207269     \\
        \hline
    \end{tabular}
\end{table}

\end{document}